\documentclass[journal]{IEEEtran}
\usepackage[pdftex]{graphicx}

\usepackage{amsmath,amssymb,amsfonts,stmaryrd}
\usepackage{bm}
\usepackage{dsfont}
\usepackage{numprint}
\newcommand{\isep}{\mathrel{{.}\,{.}}\nobreak}

\usepackage{amsthm}
\newtheorem{theorem}{Theorem}[section]

\newtheorem{lemma}[theorem]{Lemma}

\usepackage{adjustbox}
\usepackage{booktabs}
\usepackage{multirow}

\usepackage[caption=false,font=footnotesize]{subfig}
\usepackage{tikz}
\usetikzlibrary{matrix, positioning, patterns, shapes, arrows}

\usepackage{pgfplots}
\pgfplotsset{compat=newest}
\usepgfplotslibrary{groupplots}

\definecolor{Paired-1}{RGB}{31,120,180}
\definecolor{Paired-2}{RGB}{166,206,227}
\definecolor{Paired-3}{RGB}{51,160,44}
\definecolor{Paired-4}{RGB}{178,223,138}
\definecolor{Paired-5}{RGB}{227,26,28}
\definecolor{Paired-6}{RGB}{251,154,153}
\definecolor{Paired-7}{RGB}{255,127,0}
\definecolor{Paired-8}{RGB}{253,191,111}
\definecolor{Paired-9}{RGB}{106,61,154}
\definecolor{Paired-10}{RGB}{202,178,214}
\definecolor{Paired-11}{RGB}{177,89,40}
\definecolor{Paired-12}{RGB}{105,105,105}
\definecolor{Paired-13}{RGB}{80,80,80}

\usepackage[ruled, linesnumbered, vlined]{algorithm2e}

\newcommand{\furkan}[2]{\ifx&#2&{\leavevmode\color{magenta}#1}\else{\leavevmode\color{magenta}FURKAN\{}#1{\leavevmode\color{magenta}\}}\footnote{{\leavevmode\color{magenta}#2}}\PackageWarning{Furkan}{#1: #2}\fi}

\newcommand{\numberthis}{\refstepcounter{equation}\tag{\arabic{equation}}}
\newcommand{\labeln}[1]{\numberthis\label{#1}}

\newenvironment{customlegend}[1][]{%
    \begingroup
    \csname pgfplots@init@cleared@structures\endcsname
    \pgfplotsset{#1}%
}{%
    \csname pgfplots@createlegend\endcsname
    \endgroup
}%

\def\addlegendimage{\csname pgfplots@addlegendimage\endcsname}

\usepackage[draft]{fixme} 
\fxsetup{theme=color}

\AtBeginDocument{\setlength\abovedisplayskip{4pt}}
\AtBeginDocument{\setlength\belowdisplayskip{4pt}}

\ifCLASSINFOpdf
\else
\fi
\begin{document}
%
\title{High-Throughput and Energy-Efficient VLSI Architecture for Ordered Reliability Bits GRAND}
%
%
%

\author{Syed~Mohsin~Abbas,~Thibaud~Tonnellier,~Furkan~Ercan,~\IEEEmembership{Member,~IEEE}~\\Marwan~Jalaleddine~and~Warren~J.~Gross,~\IEEEmembership{Senior Member,~IEEE}
\thanks{S. M. Abbas, T. Tonnellier, M. Jalaleddine and W. J. Gross are with the Department of
Electrical and Computer Engineering, McGill University, Montr{\'e}al, Qu{\'e}bec,
Canada. F. Ercan is affiliated with the Department of Electrical and Computer Engineering, Boston University, Boston, MA, 02115. ~(email:~syed.abbas@mail.mcgill.ca,~thibaud.tonnellier@mcgill.ca,
~fercan@bu.edu,~marwan.jalaleddine@mail.mcgill.ca,~warren.gross@mcgill.ca.).~A part of this work has been presented in ICASSP 2021 \cite{ORBGRAND-VLSI}.
}}

\maketitle

\begin{abstract}

Ultra-reliable low-latency communication (URLLC), a major 5G New-Radio use case, is the key enabler for applications with strict reliability and latency requirements. These applications necessitate the use of short-length and high-rate channel codes. Guessing Random Additive Noise Decoding (GRAND) is a recently proposed Maximum Likelihood (ML) decoding technique for these short-length and high-rate codes. Rather than decoding the received vector, GRAND tries to infer the noise that corrupted the transmitted codeword during transmission through the communication channel. As a result, GRAND can decode any code, structured or unstructured. GRAND has hard-input as well as soft-input variants. Among these variants, Ordered Reliability Bits GRAND (ORBGRAND) is a soft-input variant that outperforms hard-input GRAND and is suitable for parallel hardware implementation. This work reports the first hardware architecture for ORBGRAND, which achieves an average throughput of up to $42.5$ Gbps for a code length of $128$ at a target FER of $10^{-7}$. Furthermore, the proposed hardware can be used to decode any code as long as the length and rate constraints are met. In comparison to the GRANDAB, a hard-input variant of GRAND, the proposed architecture enhances decoding performance by at least $2$ dB. When compared to the state-of-the-art fast dynamic successive cancellation flip decoder (Fast-DSCF) using a 5G polar code $(128,105)$, the proposed ORBGRAND VLSI implementation has $49\times$ higher average throughput, $32\times$ times more energy efficiency, and $5\times$ more area efficiency while maintaining similar decoding performance.

\end{abstract}

\begin{IEEEkeywords}
Area efficiency, Energy efficiency, Error Correcting Code (ECC), Guessing Random Additive Noise Decoding (GRAND), Maximum Likelihood Decoding (MLD), Ordered Reliability Bits GRAND (ORBGRAND), VLSI architecture
\end{IEEEkeywords}

%
\IEEEpeerreviewmaketitle

\begin{figure}[!t]
  \centering
  \begin{tikzpicture}[]
    \begin{groupplot}[group style={group name=fer_queries, group size= 2 by 1, horizontal sep=5pt, vertical sep=5pt},
      footnotesize,
      height=.6\columnwidth,  width=0.55\columnwidth,
      xlabel=$\frac{E_b}{N_0}$ (dB),
      xmin=0, xmax=8, xtick={0,1,...,7},
      ymode=log,
      tick align=inside,
      grid=both, grid style={gray!30},
      /pgfplots/table/ignore chars={|},
      ]

      \nextgroupplot[ylabel= FER, ytick pos=left, y label style={at={(axis description cs:-0.225,.5)},anchor=south},ymin=3e-8, ymax = 2]

      \addplot[mark=diamond*       , Paired-7 , semithick]  table[x=Eb/N0, y=FER] {data_EbNo/Polar/128_105/GRANDAB.txt};\label{gp:plot1_polar}
        \addplot[mark=square  , Paired-3 , semithick]  table[x=Eb/N0, y=FER] {data_EbNo/Polar/128_105/ORBGRAND_LW8256.txt}; \label{gp:plot2_polar}
         \addplot[mark=pentagon       , Paired-5 , thick]  table[x=Eb/N0, y=FER] {data_EbNo/Polar/128_105/DSCF_N128_K105_C11_w2_Tmax50.txt}; \label{gp:plot5_polar}
         \addplot[mark= o      , Paired-1 , semithick]  table[x=Eb/N0, y=FER] {data_EbNo/Polar/128_105/POLAR_N128_105_OSD.txt}; \label{gp:plot6_polar}
         \addplot[mark=halfcircle*  , Paired-9, semithick]  table[x=Eb/N0, y=FER] {data_EbNo/Polar/128_105/SCL_L32.txt}; \label{gp:plot7_polar}
         \addplot[mark=+,mark options={scale=1.5}  , Paired-11, semithick]  table[x=Eb/N0, y=FER] {data_EbNo/Polar/128_105/SGRAND.txt}; \label{gp:plot8_polar}

      \coordinate (top) at (rel axis cs:0,1);


      \nextgroupplot[ylabel= FER, ytick pos=right,y label style={at={(axis description cs:1.340,.5)},anchor=south},ymin=3e-8, ymax = 2]
      \addplot[mark=diamond*       , Paired-7 , semithick]  table[x=Eb/N0, y=FER] {data_EbNo/Polar/128_99/GRANDAB.txt};
        \addplot[mark=square  , Paired-3 , semithick]  table[x=Eb/N0, y=FER] {data_EbNo/Polar/128_99/ORBGRAND_LW8256.txt}; 
         \addplot[mark=pentagon       , Paired-5 , thick]  table[x=Eb/N0, y=FER] {data_EbNo/Polar/128_99/DSCF_N128_K99_C11_w2_Tmax50.txt};
         \addplot[mark= o      , Paired-1 , semithick]  table[x=Eb/N0, y=FER] {data_EbNo/Polar/128_99/POLAR_N128_99_OSD.txt};
         \addplot[mark=halfcircle*  , Paired-9, semithick]  table[x=Eb/N0, y=FER] {data_EbNo/Polar/128_99/SCL_L32.txt};
         \addplot[mark=+,mark options={scale=1.5}  , Paired-11, semithick]  table[x=Eb/N0, y=FER] {data_EbNo/Polar/128_99/SGRAND.txt};
      \coordinate (bot) at (rel axis cs:1,0);
    \end{groupplot}
    \node[below = 1cm of fer_queries c1r1.south] {\footnotesize (a) : PC (128, 105+11)};
    \node[below = 1cm of fer_queries c2r1.south] {\footnotesize (b) : PC (128, 99+11)};
    \path (top|-current bounding box.north) -- coordinate(legendpos) (bot|-current bounding box.north);
    \matrix[
    matrix of nodes,
    anchor=south,
    draw,
    inner sep=0.2em,
    draw
    ]at(legendpos)
    {
      \ref{gp:plot1_polar}& \footnotesize GRANDAB  &[1pt]
      \ref{gp:plot2_polar}& \footnotesize ORBGRAND  &[1pt]
      \ref{gp:plot5_polar}& \footnotesize DSCF \\
      \ref{gp:plot6_polar}& \footnotesize OSD  &[1pt]
      \ref{gp:plot7_polar}& \footnotesize CA-SCL (L = 32) &[1pt]
      \ref{gp:plot8_polar}& \footnotesize SGRAND \\
       }; 
  \end{tikzpicture}
  \caption{\label{fig:fer_polar_ini} Comparison of the decoding performance of different GRAND variants with OSD ($\text{Order}=2$), CA-SCL and DSCF ($\omega=2$, $T_\text{max}=50$) decoder for 5G-NR Polar Codes.}
\end{figure}
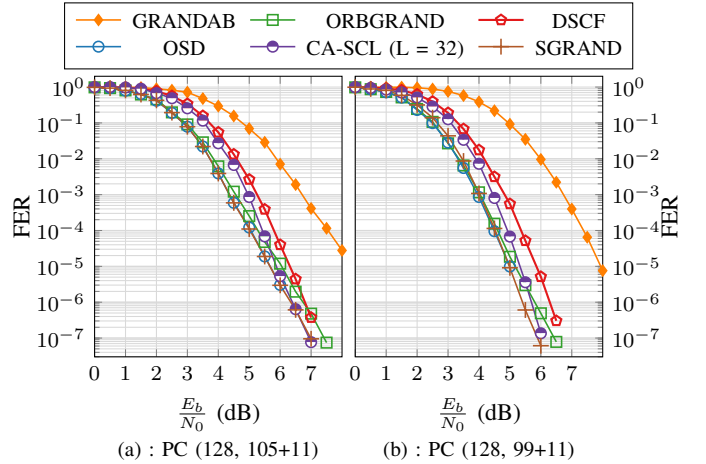

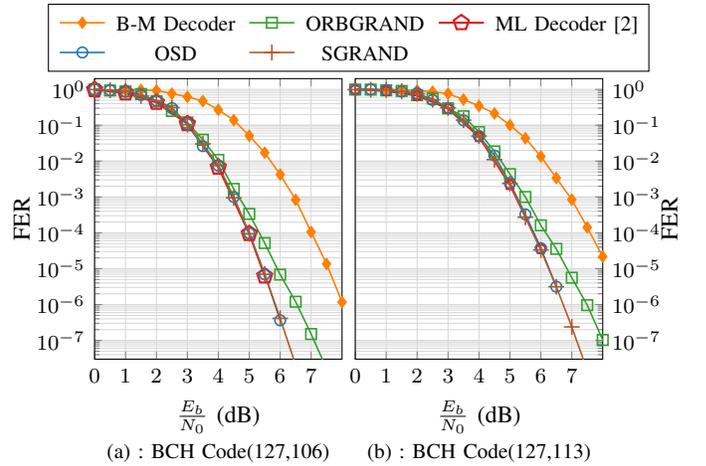
\begin{figure}[!t]
  \centering
  \begin{tikzpicture}[]
    \begin{groupplot}[group style={group name=fer_queries, group size= 2 by 1, horizontal sep=5pt, vertical sep=5pt},
      footnotesize,
      height=.6\columnwidth,  width=0.55\columnwidth,
      xlabel=$\frac{E_b}{N_0}$ (dB),
      xmin=0, xmax=8, xtick={0,1,...,7},
      ymode=log,
      tick align=inside,
      grid=both, grid style={gray!30},
      /pgfplots/table/ignore chars={|},
      ]

      \nextgroupplot[ylabel= FER, ytick pos=left, y label style={at={(axis description cs:-0.225,.5)},anchor=south},ymin=3e-8, ymax = 2]

     \addplot[mark=diamond*       , Paired-7 , semithick]  table[x=Eb/N0, y=FER] {data_EbNo/BCH/127_106/BCH_N127_106_BM.txt};\label{gp:plot1_bch}
        \addplot[mark=square  , Paired-3 , semithick]  table[x=Eb/N0, y=FER] {data_EbNo/BCH/127_106/ORBGRAND_LW8128.txt}; \label{gp:plot2_bch}
         \addplot[mark=pentagon,mark options={scale=1.5}       , Paired-5 , thick]  table[x=Eb/N0, y=FER] {data_EbNo/BCH/127_106/BCH_N127_106_ML.txt}; \label{gp:plot5_bch}
         \addplot[mark= o,mark options={scale=1}      , Paired-1 , semithick]  table[x=Eb/N0, y=FER] {data_EbNo/BCH/127_106/BCH_N127_106_OSD.txt}; \label{gp:plot6_bch}
         \addplot[mark=+,mark options={scale=1.5}   , Paired-11, semithick]  table[x=Eb/N0, y=FER] {data_EbNo/BCH/127_106/BCH_N127_106_SGRAND.txt}; \label{gp:plot8_bch}

      \coordinate (top) at (rel axis cs:0,1);


      \nextgroupplot[ylabel= FER, ytick pos=right,y label style={at={(axis description cs:1.340,.5)},anchor=south},ymin=3e-8, ymax = 2]
      \addplot[mark=diamond*       , Paired-7 , semithick]  table[x=Eb/N0, y=FER] {data_EbNo/BCH/127_113/BCH_N127_113_BM.txt};
        \addplot[mark=square  , Paired-3 , semithick]  table[x=Eb/N0, y=FER] {data_EbNo/BCH/127_113/ORBGRAND_LW8128.txt}; 
         \addplot[mark=pentagon,       , Paired-5 , thick]  table[x=Eb/N0, y=FER] {data_EbNo/BCH/127_113/BCH_N127_113_ML.txt}; 
         \addplot[mark= o,mark options={scale=1}      , Paired-1 , semithick]  table[x=Eb/N0, y=FER] {data_EbNo/BCH/127_113/BCH_N127_113_OSD.txt}; 

         \addplot[mark=+,mark options={scale=1.5} , Paired-11, semithick]  table[x=Eb/N0, y=FER] {data_EbNo/BCH/127_113/BCH_N127_113_SGRAND.txt}; 

      \coordinate (bot) at (rel axis cs:1,0);
    \end{groupplot}
    \node[below = 1cm of fer_queries c1r1.south] {\footnotesize (a) : BCH Code(127,106)};
    \node[below = 1cm of fer_queries c2r1.south] {\footnotesize (b) : BCH Code(127,113)};
    \path (top|-current bounding box.north) -- coordinate(legendpos) (bot|-current bounding box.north);
    \matrix[
    matrix of nodes,
    anchor=south,
    draw,
    inner sep=0.2em,
    draw
    ]at(legendpos)
    {
      \ref{gp:plot1_bch}& \footnotesize B-M Decoder &[1pt]
      \ref{gp:plot2_bch}& \footnotesize ORBGRAND  &[1pt]
      \ref{gp:plot5_bch}& \footnotesize ML Decoder \cite{kaiserslautern} \\
      \ref{gp:plot6_bch}& \footnotesize OSD &[1pt]
      \ref{gp:plot8_bch}& \footnotesize SGRAND \\
       }; 
  \end{tikzpicture}
  \caption{\label{fig:fer_bch_ini} Comparison of the decoding performance of different GRAND variants with OSD ($\text{Order}=2$) and ML decoder for BCH Codes.}
\end{figure}

\section{Introduction}

\IEEEPARstart{F}{ollowing} Shannon's seminal 1948 paper \cite{Shannon48}, much work was directed toward developing practical coding schemes that could approach channel capacity (named ``the Shannon limit''). Early proposed channel coding techniques mainly focused on designing error correcting codes where the number of correctable errors is selected by design. This approach towards channel coding brought rise to Hamming codes \cite{Hamming50} and Bose–Chaudhuri–Hocquenghem (BCH) codes \cite{Hocquenghem59,Bose1960}. Recently, researchers have been trying to discover new capacity achieving and capacity-approaching codes. Most notably, Turbo codes \cite{turbo} and LDPC codes \cite{Gallager62}, were proposed over time as capacity-approaching codes. On the other hand, Polar codes \cite{Arikan09}, proposed in 2008, are the first proven class of codes that asymptotically reach the Shannon limit for binary-input symmetric memory-less channels, as well as discrete and continuous memory-less channels \cite{Arikan2}. However, designing codes that perform well in the short-to-medium block length regime is challenging~\cite{cocskun2019efficient}.

The Guessing Random Additive Noise Decoding (GRAND) \cite{Duffy19TIT} is a recently proposed Maximum Likelihood (ML) decoding technique for short-length and high-rate channel codes. Emerging applications such as intelligent transportation systems (ITS) \cite{URLLC1}, the internet of things (IoT) \cite{IoT1,IoT2}, augmented and virtual reality, machine to machine communication (M2M) \cite{URLLC3}, and ultra-reliable low-latency communication (URLLC) \cite{URLLC2} (5G New-Radio (NR) use case) necessitate short data packets with high reliability (target FER of $10^{-5}\sim10^{-9}$) and ultra low latency to deal with critical events. GRAND is a maximum likelihood decoder for these channel codes with short lengths and high rates. Furthermore, because GRAND focuses on noise, transmissions with low noise are quickly decoded. GRAND is therefore suitable for applications that require high reliability and ultra-low latency.

In addition to GRAND, other universal decoders for $(n,k)$ linear block codes, where $n$ is the code length and $k$ is the code dimension, include brute-force ML decoding, Ordered Statistic Decoding (OSD) \cite{Fossorier95} and its variants \cite{OSD3,OSD4,OSD5,OSD6}. Due to the high complexity requirement, the ML decoder is unsuitable for high-rate codes. Whereas the OSD and its variants permute the columns of the generator matrix ($\bm{G}$) of the underlying code in a way that is dependent on the received vector of channel observation values ($\bm{y}$), and the $\bm{G}$ matrix is transformed to a systematic form using Gaussian elimination. OSD is unsuitable for parallel hardware implementation \cite{Scholl13} due to the complexity of Gaussian elimination ($\mathcal{O}(n^3)$) and the reliance of column permutations of $\bm{G}$ on the received vector ($\bm{y}$). GRAND, on the other hand, provides a low-complexity ML decoding solution for short-length and high-rate codes because it does not require Gaussian elimination or $\bm{G}$ column permutations. GRAND instead requires simple bit-flipping and syndrome check (codebook membership verification) operations. 

GRAND generates test error patterns, which are successively applied to the vector of channel observation values ($\bm{y}$) and the resulting vector is evaluated for codebook membership. GRAND can be used with any codebook as long as there is a method for validating a vector's codebook membership. For linear codebooks ($\mathcal{C}$), the codebook membership for a vector can be verified using the parity check matrix $\bm{H}$. For other non-structured codebooks, stored in a dictionary, the codebook membership of a vector can be checked with a dictionary lookup. For the rest of the discussion, we restrict ourselves to $(n,k)$ linear block codes.

The order in which these test error patterns are generated is the primary distinction between GRAND variants. GRAND with ABandonment (GRANDAB) \cite{Duffy19TIT} is a hard decision input version of GRAND that produces test error patterns in increasing Hamming weight order up to weight $AB$. Ordered Reliability Bits GRAND (ORBGRAND) \cite{duffy2020ordered} and Soft GRAND (SGRAND) \cite{solomon2020soft} are soft-input variants that efficiently utilize soft information (channel observation values), resulting in improved decoding performance over hard-input GRANDAB. 

Figure \ref{fig:fer_polar_ini} compares the decoding performance of various GRAND variants for decoding 5G NR CRC-aided polar code (128,105+11) and polar code (128,99+11). Furthermore, the decoding performance of state-of-the-art soft-input decoders such as the CRC-Aided Successive Cancellation List (CA-SCL) decoder \cite{Tal15,LLR-List}, Dynamic SC-Flip (DSCF) \cite{Chandesris,Ercan_2020} decoder and OSD \cite{Fossorier95,OSD3,OSD4} is included for reference. The numerical simulation results presented in this work are based on BPSK modulation over an AWGN channel. The ORBGRAND and SGRAND soft-input decoders outperform the hard-input GRANDAB variant in decoding performance, and the SGRAND variant achieves ML decoding performance similar to OSD, as shown in Fig. \ref{fig:fer_polar_ini}. Figure \ref{fig:fer_bch_ini} compares the decoding performance of different GRAND variants with OSD and ML decoding of BCH codes $(127,106)$ and $(127,113)$, respectively. The ML decoding results are obtained from \cite{kaiserslautern}. Similar trends in decoding performance can be seen in the BCH codes depicted in Fig. \ref{fig:fer_bch_ini}, where soft-input variants of GRAND (ORBGRAND and SGRAND) outperform the hard-input traditional Berlekamp-Massey (B-M)  \cite{Berlekamp68,Massey69} decoder and the SGRAND achieves ML performance comparable to OSD.

The SGRAND outperforms the other GRAND variants in terms of decoding performance; however, SGRAND is not suitable for parallel hardware implementation. The generated test error patterns in SGRAND are interdependent, and their query order changes for each received vector of channel observation values ($\bm{y}$) \cite{solomon2020soft}. As a result, SGRAND does not lend itself to efficient parallel hardware implementation, and a sequential hardware implementation will result in a high decoding latency, which is unsuitable for applications that require ultra-low latency. The ORBGRAND, on the other hand, generates test error patterns in a predetermined logistic weight order based on integer partitioning. Furthermore, the test error patterns generated are mutually independent and can be generated in parallel. ORBGRAND is thus highly parallelizable and well suited to parallel hardware implementation.

In this paper, we investigate the parameters that affect the decoding performance as well as the computational complexity of the ORBGRAND algorithm, and we propose modifications to the ORBGRAND algorithm to aid hardware implementation and reduce complexity. In addition, we present a high-throughput and energy-efficient ORBGRAND VLSI architecture. The VLSI implemenation results show that, considering a code of length $128$ and a target FER of $10^{-7}$, the proposed harware architecture can achieve an average information throughput of up to $42.5$ Gbps. Furthermore, the proposed hardware can be used to decode any code as long as the length and rate constraints are met. When compared to the state-of-the-art fast dynamic successive cancellation flip decoder (Fast-DSCF) \cite{Chandesris,Ercan_2020} using a 5G NR polar code $(128,105)$, the proposed ORBGRAND implementation results in $49\times$ more average throughput, $32\times$ more energy efficiency, and $5\times$ more area efficiency.

It should be noted that a part of this work was previously discussed in \cite{ORBGRAND-VLSI}. This paper builds on the earlier work \cite{ORBGRAND-VLSI} and extends the proposed ORBGRAND in following ways:

\begin{itemize}
\item For various classes of channel codes, including Bose-Chaudhuri-Hocquenghem (BCH) codes, Cyclic Redundancy Check (CRC) codes, Random Linear Codes (RLCs), and Polar codes, the proposed ORBGRAND is evaluated in terms of decoding performance and compared with state-of-the-art decoding approaches.
\item The ORBGRAND VLSI architecture \cite{ORBGRAND-VLSI} can only support a limited set of parameters ($LW \leq 64$ and $P \leq 6$). The proposed ORBGRAND VLSI architecture is expanded to support the additional parameters ($LW \leq 96$ and $P \leq 8$) in this work. Furthermore, detailed VLSI implementation results for area, power, throughput, hardware efficiency, and energy efficiency are presented.
\item This paper describes the proposed test error pattern generation scheme for the ORBGRAND hardware and presents a step-by-step procedure for leveraging shift registers to generate error patterns corresponding to logistic weight order based on the integer partitioning procedure.
\item This paper provides a comprehensive analysis of worst-case latency and throughput for selecting different parameters for the proposed ORBGRAND hardware. This analysis assists in the selection of optimal ORBGRAND hardware parameters to meet the throughput and latency requirements of a target application.
\item For the proposed ORBGRAND VLSI hardware, we propose segmenting the sorter module into multiple partitions. The effect of partition number on the displacement of LLR elements $\vert\bm{y}_i\vert$ ($\forall i \in [1,n]$) from their correct locations and their effect on ORBGRAND decoding performance is thoroughly investigated. Finally, the ORBGRAND with segmented sorter approach is implemented and compared to the baseline ORBGRAND hardware with non-segmented sorter, and the effect of varying the number of sorter-segments on area overhead is investigated. 
\end{itemize}

The remainder of this work is structured as follows: Section II includes preliminary information on GRAND and ORBGRAND. Section III investigates ORBGRAND parameters and proposes modifications for simple hardware implementation as well as complexity reduction of ORBGRAND. The proposed hardware architecture for ORBGRAND is detailed in Section IV. Section V presents implementation results and compares them to the Fast-DSCF and GRANDAB decoders. Finally, in Section VI, concluding remarks are made.  

\section{Preliminaries}
\subsection{Notations}
Matrices are denoted by a bold upper-case letter ($\bm{M}$), while vectors are denoted with bold lower-case letters ($\bm{v}$).
The transpose operator is represented by $^\top$. The number of $k$-combinations from a given set of $n$ elements is noted by $\binom{n}{k}$. $\mathds{1}_n$ is the indicator vector where all locations except the $n^{\text{th}}$ are $0$ and the $n^{\text{th}}$ is $1$. Similarly, $\mathds{1}_{\bm{v}}$ is the indicator vector in which all locations $v_{i}$ $(\forall i \in [1, n])$ are $1$. All the indices start at $1$. For this work, all operations are restricted to the Galois field with 2 elements, noted $\mathbb{F}_2$. The symbols $\therefore$ and $\because$ denote \textit{therefore} and \textit{because} respectively. $\bm{a}\circ\bm{b}$ denotes permuting elements in $\bm{a}$ according to the permutation order in $\bm{b}$. 

\subsection{GRAND decoding of linear block codes}
A linear block code is a linear mapping $g: \mathbb{F}_2^k \rightarrow \mathbb{F}_2^n$, 
where $k < n$. In this mapping, a vector $\bm{u}$ of size $k$ maps to a vector $\bm{c}$ of size $n$ and the ratio $R \triangleq \frac{k}{n}$ is called the code-rate. For every linear block code, there exists a $k \times n$ matrix $\bm{G}$
called generator matrix and a $(n-k) \times n$ matrix $\bm{H}$ called parity check matrix. The set of the $2^k$ vectors $\bm{c}$ is called a code $\mathcal{C}$, whose elements $\bm{c}$ are called codewords and each codeword verifies the following property:
\begin{equation}
\forall~\bm{c} \in \mathcal{C},~\bm{H}\cdot\bm{c}^\top = \bm{0}.
\label{eq:pcheck}
\end{equation}
Consider the case where $\bm{c}$ was sent via a noisy channel and $\bm{r}$ was received at the channel's output. Because of the channel noise, $\bm{r}$ and $\bm{c}$ might differ. As a result, the relationship between $\bm{r}$ and $\bm{c}$ may be deduced as follows: $\bm{r} = \bm{c}~\oplus~\bar{\bm{e}}$, where $\bar{\bm{e}}$  represents the channel-induced noise.

GRAND sequentially generates the test error patterns ($\bm{e}$), and applies them to $\bm{r}$ and checks for codebook membership of $\bm{r}$ by verifying that
\begin{equation}
\bm{H} \cdot(\bm{r} \oplus \bm{e})^\top
\label{eq:constraint}
\end{equation}
is equal to zero. If so, $\bm{e}$ is the guessed noise and $\hat{\bm{c}} \triangleq \bm{r}~\oplus~\bm{e}$ is the estimated codeword. GRAND's focus is on noise, thus it can be used with any codebook as long as there is a method for validating a vector's codebook membership. For linear codebooks, the codebook membership for a vector can be verified using $\bm{H}$. 

\subsection{ORBGRAND decoding}\label{sec:ORBGRANDdecoding}
Algorithm \ref{alg:ORBgrand} summarizes the steps of the ORBGRAND. The inputs to the algorithm are the vector of channel observation values (log-likelihood ratios (LLRs)) $\bm{y}$ of size $n$, a $(n-k)\times n$ parity check matrix of the code $\bm{H}$, an $n\times k$ matrix $\bm{G}^{-1}$ such that $\bm{G}\cdot \bm{G}^{-1}$ is the $k\times k$ identity matrix, with $\bm{G}$ a generator matrix of the code, and the maximum logistic weight considered $LW_\text{max}$ ($LW_\text{max} \leq \frac{n(n+1)}{2}$). The logistic weight ($LW$) corresponds to the sum of the indices of non zero elements in the test error patterns ($\bm{e}$)\cite{duffy2020ordered}. For example, $\bm{e} = [0,1,0,0,1,0]$ has a Hamming weight of $2$, whereas the logistic weight is $2 + 5 = 7$.

ORBGRAND begins by sorting $\bm{y}$ in ascending order of absolute value of LLRs ($\vert\bm{y}\vert$), and the corresponding indices are recorded in a permutation vector denoted by $\bm{ind}$ (line 1). Following that, all integer partitions ($\bm{\lambda} = (\lambda_1, \lambda_2, \ldots, \lambda_P) \vdash i$ where $P\in[1,P_\text{max}]$ and $i\in[0,LW_\text{max}]$; explained in section \ref{sec:integerPartition} ) are generated for each logistic weight (line 3). The integer partition ($\bm{\lambda}$) is used to generate a test error pattern ($\bm{e}$), which is then ordered using the permutation vector $\bm{ind}$ (line 5-6). The generated test error patterns are then applied sequentially to the hard decision vector ($\hat{\bm{y}}$), which is obtained from the input soft channel observation values ($\bm{y}$). The resulting vector is then queried for codebook membership (line 7). If the codebook membership criterion (\ref{eq:constraint}) is met, then $\bm{e}$ is the guessed noise and $\hat{\bm{c}} \triangleq \hat{\bm{y}}~\oplus~\bm{e}$ is the estimated codeword. Otherwise, either larger logistic weights or the remaining error patterns for that logistic weight are considered. Finally, using $\bm{G}^{-1}$ (line 8), the original message ($\hat{\bm{u}}$) is retrieved from the estimated codeword, and the decoding process is terminated.

\begin{algorithm}[t]
\caption{\label{alg:ORBgrand}ORBGRAND Algorithm}
    \DontPrintSemicolon
    \SetAlgoVlined  
    \SetKwData{e}{$\bm{e}$}
    \SetKwData{edash}{$\bar{\bm{e}}$}
    \SetKwData{p}{$\bm{p}$}
    \SetKwData{s}{$\bm{S}$}
    \SetKwData{ind}{$\bm{ind}$}
    \SetKwData{LLRsor}{$\bar{\bm{y}}$}
    \SetKwData{LLR}{$\bm{y}$}
    \SetKwData{LLRabs}{$\lvert\bm{y}\rvert$}
    \SetKwData{sortSet}{$[\bm{r},\bm{ind}]$}
    \SetKwData{estm}{$\hat{\bm{u}}$}
    \SetKwData{ginv}{$\bm{G}^{-1}$}
    \SetKwData{LW}{${LW_\text{max}}$}
    \SetKwData{HW}{$\bm{HW}$}
    \SetKwData{yhat}{$\hat{\bm{y}}$}
    \KwIn{\LLR, $\bm{H}$, \ginv, \LW}
    \KwOut{\estm}
    \SetKwFunction{RecursiveComputeLLRs}{recursiveComputeLLRs}
    \SetKwFunction{DecodeRONE}{decodeR1}
    \SetKwFunction{RDecodeRONE}{redecodeR1}
    \SetKwFunction{DecodeRZERO}{decodeR0}
    \SetKwFunction{Find}{findCandidate}
    \SetKwFunction{new}{generateErrorPattern}
    \SetKwFunction{intPartition}{$(\lambda_1, \lambda_2, \ldots, \lambda_P) \vdash i$}
    \SetKwFunction{Sort}{Sort}
    $[\ind,\LLRsor] \leftarrow$ \Sort{\LLRabs} \tcp*[r]{$\bar{\bm{y}}_i\leq\bar{\bm{y}}_j~$$~(\forall i < j)$}
    \For{$i \gets 0$ to \LW}{
        $\s \leftarrow$ \intPartition \tcp*[r]{$ P \in [1,P_\text{max}]$)}
        \ForAll{$\bm{j}$ in \s}{
          $\e \leftarrow \bm{0}$\;
          $\e \leftarrow (\e \oplus \mathds{1}_{\bm{j}})\circ\ind$\; 
          \If{$\bm{H} \cdot(\yhat \oplus \e)^\top == \bm{0}$} {
            $\estm \leftarrow (\yhat \oplus \e)\cdot\ginv$\;
            \KwRet{\estm}
            }
        }
    }
\end{algorithm}

\begin{figure}[!t]
  \centering
  \begin{tikzpicture}[]
    \begin{groupplot}[group style={group name=fer_queries, group size= 2 by 1, horizontal sep=5pt, vertical sep=5pt},
      footnotesize,
      height=.6\columnwidth,  width=0.55\columnwidth,
      xlabel=$\frac{E_b}{N_0}$ (dB),
      xmin=0, xmax=8, xtick={0,1,...,7},
      ymode=log,
      tick align=inside,
      grid=both, grid style={gray!30},
      /pgfplots/table/ignore chars={|},
      ]

      \nextgroupplot[ylabel= FER, ytick pos=left, y label style={at={(axis description cs:-0.225,.5)},anchor=south},ymin=3e-8, ymax = 2]

      \addplot[mark=diamond       , Paired-1 , semithick]  table[x=Eb/N0, y=FER] {data_EbNo/CRC/128_104/CRC_N128_104_GRANDAB.txt};\label{gp:plot1_crc}
        \addplot[mark=square  , Paired-3 , semithick]  table[x=Eb/N0, y=FER] {data_EbNo/CRC/128_104/CRC_N128_104_ORBGRAND.txt}; \label{gp:plot2_crc}

      \coordinate (top) at (rel axis cs:0,1);


      \nextgroupplot[ylabel= FER, ytick pos=right,y label style={at={(axis description cs:1.340,.5)},anchor=south},ymin=3e-8, ymax = 2]
      \addplot[mark=diamond       , Paired-1 , semithick]  table[x=Eb/N0, y=FER] {data_EbNo/RLC/128_104/GRANDAB.txt};
        \addplot[mark=square  , Paired-3 , semithick]  table[x=Eb/N0, y=FER] {data_EbNo/RLC/128_104/ORBGRAND.txt}; 

      \coordinate (bot) at (rel axis cs:1,0);
    \end{groupplot}
    \node[below = 1cm of fer_queries c1r1.south] {\footnotesize (a) : CRC Code};
    \node[below = 1cm of fer_queries c2r1.south] {\footnotesize (b) : RLC Code};
    \path (top|-current bounding box.north) -- coordinate(legendpos) (bot|-current bounding box.north);
    \matrix[
    matrix of nodes,
    anchor=south,
    draw,
    inner sep=0.2em,
    draw
    ]at(legendpos)
    {
      \ref{gp:plot1_crc}& \footnotesize  GRANDAB ($\text{AB}=3$)  &[1pt]
      \ref{gp:plot2_crc}& \footnotesize  ORBGRAND  \\
       }; 
  \end{tikzpicture}
  \caption{\label{fig:fer_crc_rlc} Comparison of the GRANDAB and ORBGRAND decoding performance using (a) Cyclic Redundancy Check (CRC) codes and (b) Random Linear Codes (RLCs) for $n=128$ and $k=104$.}
\end{figure}
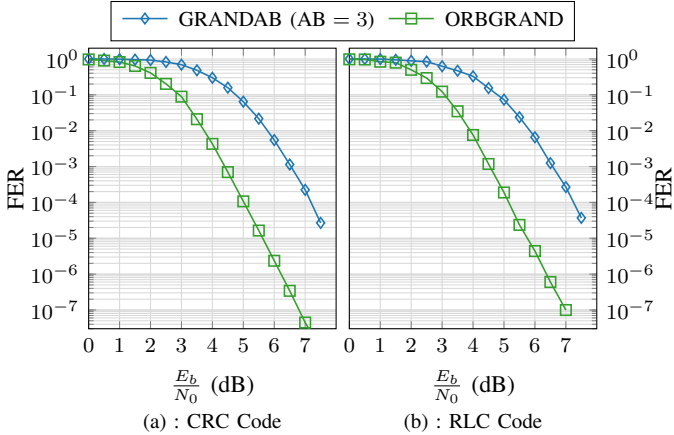

The frame error rate (FER) performance of ORBGRAND, a soft decision decoder, is compared to a hard decision variant GRANDAB for cyclic redundancy check (CRC) codes \cite{Peterson61} and Random Linear Codes (RLCs) \cite{RLC1,RLC2} in Fig. \ref{fig:fer_crc_rlc}. CRC codes \cite{Peterson61}  are typically used to detect errors in communication systems and to assist list-based channel code decoders in selecting the final candidate codeword. On the other hand, CRC codes can also be used for error correction using the GRAND algorithm. The concept of using CRC codes for error correction with GRAND decoding was presented in \cite{GRANDAB-VLSI} and expanded on in \cite{CRCGrand}. 
 RLCs \cite{RLC1,RLC2} are linear block codes that are theoretically high-performing \cite{RLC1,RLC2} but are not considered realistic in terms of decodability. For CRC code (128,104), generator polynomial \texttt{0xB2B117} is used. As seen in Fig. \ref{fig:fer_crc_rlc}, ORBGRAND outperforms GRANDAB by at least $2$ dB for a target FER $\geq$ $10^{-5}$ for both CRC codes and RLCs.

 To conclude, both GRAND and ORBGRAND can be used to decode any linear block code ($n,k$), structured or unstructured, as long as the underlying code's parity check matrix ($\bm{H}$) is provided. Furthermore, as a soft-input decoder, ORBGRAND outperforms hard-input GRANDAB.

\section{ORBGRAND Design Considerations}\label{sec:integerPartition}
ORBGRAND is centered around generating distinct integer partition of a particular logistic weight ($LW$), and these integer partitions are then used to generate test error patterns ($\bm{e}$). 

An integer partition $\bm{\lambda}$ of a positive integer $m$, noted $\bm{\lambda} = (\lambda_1, \lambda_2, \ldots, \lambda_P) \vdash m$ where $\lambda_1>\lambda_2>\ldots>\lambda_P$, is the multiset of positive integers $\lambda_i$ $(\forall i \in [1, P])$ that sum to $m$. If all parts $\lambda_i$ $(\forall i \in [1, P])$ of the integer partition are different, the partition is called distinct. Note that the generated test error pattern obtained from an integer partition with $P$ elements has a Hamming weight of $P$. The ORBGRAND considers $\frac{n(n+1)}{2}$ as the maximum logistic weight for a $(n,k)$ linear block code ($LW_\text{max} = \frac{n(n+1)}{2}$). Furthermore, the generated TEPs have a maximum Hamming weight of $n$ ($HW_\text{max} = n$). It should be noted that only distinct integer partitions are taken into account when generating TEPs, and all parts ($\lambda_i$) of the partition are less than or equal to $n$ ($\lambda_i\leq~n$ $(\forall i \in [1, P])$).

\subsection{Parametric analysis of ORBGRAND}

As seen in Algorithm \ref{alg:ORBgrand}, $LW_\text{max}$ is an important parameter for the ORBGRAND. $LW_\text{max}$ impacts ORBGRAND's decoding performance as well as its computational complexity. The computational complexity of GRAND and its variants can be expressed in terms of the number of codebook membership queries required. Furthermore, the complexity can be further subdivided into worst-case complexity, which corresponds to the maximum number of codebook membership queries required, and average complexity, which corresponds to the average number of codebook membership queries required. It should be noted that with improved channel conditions, the average complexity of GRAND and its variants decreases sharply as transmissions subject to light noise are quickly decoded \cite{Duffy19TIT}\cite{solomon2020soft}\cite{duffy2020ordered}.

In addition to $LW_\text{max}$, another important parameter of ORBGRAND is the number of elements ($P$) in the generated integer partition ($\bm{\lambda}$). Furthermore, $P$ denotes the Hamming weight of the generated test error pattern. ORBGRAND's $LW_\text{max}$ and $P$ parameters can be appropriately chosen to reduce worst-case complexity while having a minimal impact on decoding performance.

Fig. \ref{fig:fer_polar_bch}(a) depicts the impact of different parameters ($LW_\text{max}$, $P$) on ORBGRAND decoding performance for 5G CRC-aided polar code (128,105+11) \cite{Duffy20205g,3GPP} with BPSK modulation over an AWGN channel. Furthermore, the decoding performance of state-of-the-art soft-input decoders such as CA-SCL decoder \cite{Tal15,LLR-List} and DSCF \cite{Chandesris,Ercan_2020} decoder is included for reference. The number of DSCF attempts ($T_\text{max}$) parameter is set to 50, and the maximum bit-flipping order ($\omega$) is set to 2. As seen in Fig. \ref{fig:fer_polar_bch}(a), the FER performance of ORBGRAND outperforms the GRANDAB (AB=3) \cite{Duffy19TIT} decoder by at least $2$ dB for target FERs $\leq 10^{-5}$. Furthermore, the ORBGRAND decoder outperforms the CA-SCL decoder \cite{Tal15,LLR-List}  and the DSCF decoder \cite{Chandesris,Ercan_2020} for decoding polar code (128,105+11) at target FER of $10^{-4}$ and $10^{-6}$, respectively.

The maximum number of codebook membership queries for $LW_\text{max}$ values of 128, 96, and 64 is $5.33\times10^{7}$, $3.69\times10^{6}$, and $1.5\times10^{5}$, respectively. When $LW_\text{max}$ is decreased from $128$ to $64$, a performance degradation of $0.2$ dB is observed at FER = $10^{-7}$, as shown in Fig. \ref{fig:fer_polar_bch}(a). The complexity, on the other hand, is reduced by $355\times$ as a result of this reduction. Similarly, the degradation in ORBGRAND decoding performance for the considered polar code with $LW_\text{max}=64$  is negligible when $P=6$ is used instead of an unbounded $P$. As a consequence, with $LW_\text{max}=64$ and $P=6$, the maximum number of queries is limited to $1.16\times10^{5}$, and the ORBGRAND (for target FER of $10^{-5}$) performs similarly to the state-of-the-art DSCF polar code decoder  \cite{Chandesris,Ercan_2020}. 

Similarly, Fig. \ref{fig:fer_polar_bch}(b) presents a comparison of ORBGRAND decoding performance for the BCH (127,106) code. The ORBGRAND decoding performance is compared to that of the hard decision input Berlekamp-Massey (B-M) decoder \cite{Berlekamp68,Massey69} and the soft-input OSD \cite{Fossorier95,OSD3,OSD4} decoder. ORBGRAND decoding of BCH (127, 106) code results in a $1.7$dB performance gain at a target FER of $10^{-5}$ when compared to B-M decoder. For target FER of $10^{-6}$, the OSD decoder outperforms the ORBGRAND decoder by $0.7$dB. However, due to the complexity of Gaussian elimination ($\mathcal{O}(n^3)$) and the reliance of column permutation of the $\bm{G}$ matrix on the received vector from the channel ($\bm{y}$), OSD is unsuitable for parallel hardware implementation \cite{Scholl13}. ORBGRAND, on the other hand, requires only simple bit-flipping and syndrome check (codebook membership verification) operations, making it ideal for applications requiring ultra-low decoding latency.

To conclude, the appropriate selection of ORBGRAND parameters $LW_\text{max}$ and $P$ results not only in a reduction in computational complexity but also in the design of simple hardware, as seen in the section \ref{sec:vlsi}.

\subsection{Proposed simplified generation of integer partitions ($\bm{\lambda}$)}
A hardware implementation for the generation of integer partitions was proposed in \cite{T14high-speedhardware}. However, since the generated partitions are not distinct, their approach cannot be directly applied to our proposed ORBGRAND architecture. Furthermore, their integer partition generation is sequential, which is unsuitable for use in a parallelized, high-throughput hardware architecture. In this section, we propose a method for generating integer partitions of a particular logistic weight using a specific arrangement of shift registers and XOR gates.

For generating integer partitions of a specific logistic weight $m$, we noticed that a breakdown of $m$ generates convenient patterns. For example, for $m = 12$ the distinct integer partitions are $\bm{\lambda} = \{(12);$ $(11,1);$ $ (10,2);$ $(9,3);$ $(8,4);$ $(7,5);$ $(9,2,1);$ $(8,3,1);$ $(7,4,1);$ $(6,5,1);$ $(7,3,2);$ $(6,4,2);$ $(5,4,3);$ $ (6,3,2,1);$ $(5,4,2,1)\}$. If a listing order is followed for $P =2$ (i.e. the subset $\{(11,1);$ $(10,2);$ $(9,3);$ $(8,4);$ $(7,5)\}$), the first integer descends while the second ascends. Therefore for a particular logistic weight $m$, integer partitions of size 2 ($P =2$) can be generated as $(\lambda_1,\lambda_2)\vdash m$ where $\lambda_{2}\in\mathcal[1,\left\lfloor\frac{m}{2}\right\rfloor-1 ]$ and $\lambda_{1} = m-\lambda_{2}$.

Similar trends can be observed for higher-order partitions such as $P=3$ (i.e. the subset $\{(9,2,1);$ $(8,3,1);$ $(7,4,1);$ $(6,5,1);$ $(7,3,2);$ $(6,4,2);$ $(5,4,3)\}$), the first integer descends while the second ascends as the third integer remains fixed until all iterations for the first two integers are complete. Hence, integer partitions of size 3 ($P =3$) can be generated as $(\lambda_{1},\lambda_{2},\lambda_{3})$$\vdash$$m$ where, $\lambda_{3}\in\mathcal[1,\lambda_{3}^{max}]$, $\lambda_{2}\in\mathcal[\lambda_{3}+1,\lambda_{2,\lambda_{3}}^{max}]$ and $\lambda_{1} = m-\lambda_{2}-\lambda_{3}$. Moreover, $\lambda_{3}^{max}$ is the maximum value of $\lambda_{3}$, and $\lambda_{2,\lambda_{3}}^{max}$ is the maximum value of $\lambda_{2}$ for a specific value of $\lambda_{3}$ ($\lambda_{3}\in\mathcal[1,\lambda_{3}^{max}]$). 

 In general, an integer partition of size $P$ can be generated as $(\lambda_1, \lambda_2, \ldots, \lambda_P) \vdash m$ where $\lambda_P \in [1,\lambda_{P}^\text{max}]$, $\lambda_i \in [\lambda_{i+1}+1,\lambda_{i,\lambda_{i+1},\ldots,\lambda_{P-1}}^\text{max}] \forall i \in [2, P-1]$ and $\lambda_1=m-\sum\limits_{i=2}^{P}\lambda_{i}$. Moreover, $\lambda_{P}^{max}$ is the maximum value of $\lambda_{P}$, and $\lambda_{i,\lambda_{i+1}\ldots,\lambda_{P-1}}^\text{max}$ is the maximum value of $\lambda_{i}$ for specific values of $\lambda_{j}$ $(\forall j \in [i+1, P-1])$. For similpicity, we will denote $\lambda_{i,\lambda_{i+1}\ldots,\lambda_{P-1}}^\text{max}$  as $\lambda_{i}^\text{max}$. The maximum value for each $\lambda_i$ $(\forall i \in [2, P])$ is bounded by (\ref{eq:lambda_bound}).

\begin{lemma}
If a positive integer $m$ is partitioned into $P$ distinct parts, $\forall i \in [2, P]$, and assuming that $\lambda_i$ are 
ordered, the maximum value for each $\lambda_i$ is bounded by
\begin{equation}
 \lambda_i^\text{max} < \frac{2\times m - (i\times(i-1))+2-2\times\sum\limits_{j=i+1}^{P}\lambda_{j}}{2\times i}.
\label{eq:lambda_bound}    
\end{equation} whereas the first value of $\lambda$ is given as $\lambda_1=m-\sum\limits_{j=2}^{P}\lambda_{j}$.
\end{lemma}

\begin{proof}
The proof is provided in Appendix A.
\end{proof}

\begin{figure*}[!t]
    \centering
    \captionsetup[subfigure]{oneside}
    \begin{tikzpicture}

    \begin{customlegend}[legend columns=4,legend style={align=left,draw},legend entries={{} {{\footnotesize{ORBGRAND}, $LW_\text{max}$=8256}},{} {{\footnotesize{ORBGRAND}, $LW_\text{max}$=8128}},{} {{\footnotesize{ORBGRAND}, $LW_\text{max}$=128}},{} {{\footnotesize{ORBGRAND}, $LW_\text{max}$=127}},{} {{\footnotesize{ORBGRAND}, $LW_\text{max}$=96}},{} {{\footnotesize{ORBGRAND}, $LW_\text{max}$=96, $P \leq 8$}},{} {{\footnotesize{ORBGRAND}, $LW_\text{max}$=64}},
    {} {{\footnotesize{ORBGRAND}, $LW_\text{max}$=64, $P \leq 6$}},{} {{\footnotesize{GRANDAB }, $AB=3$}}, {} {{\footnotesize{B-M Decoder  }}},{} {{\footnotesize{DSCF, $\omega=2$, $T_\text{max}=50$}}},{} {{\footnotesize{CA-SCL ($L = 32$)}}},{} {{\footnotesize{OSD ($\text{Order}=2$)}}}}]
    \addlegendimage{draw=Paired-1,mark=o, semithick}
    \addlegendimage{draw=Paired-6,mark=o, semithick}
    \addlegendimage{draw=Paired-4,mark=diamond, semithick}
    \addlegendimage{draw=Paired-8,mark=diamond, semithick}
    \addlegendimage{draw=Paired-12,mark=star,semithick}
    \addlegendimage{draw=Paired-3,mark=square, semithick}
    \addlegendimage{draw=Paired-5,mark=triangle, semithick}
    \addlegendimage{draw=Paired-7,mark=pentagon, semithick}
    \addlegendimage{draw=Paired-11,mark=*,fill=Paired-11, semithick, dashed} 
    \addlegendimage{draw=Paired-13,mark=*,fill=Paired-13,semithick, dashed}
    \addlegendimage{draw=Paired-9, thick, dashed}
    \addlegendimage{draw=Paired-11, mark=+,mark options={scale=1.2}, semithick}
    \addlegendimage{draw=Paired-13,mark=pentagon, semithick} 
    \end{customlegend}
    \end{tikzpicture}
    \subfloat[Polar Code(128,105+11)]{
    \begin{tikzpicture}
    \begin{semilogyaxis}[
            footnotesize, width=\columnwidth, height=.67\columnwidth,    
            xmin=0, xmax=8, xtick={0,1,...,7},
            ymin=3e-8,  ymax=2,
            xlabel=$\frac{E_b}{N_0}$, ylabel=FER,  
            grid=both, grid style={gray!30},
            tick align=outside, tickpos=left, 
            legend pos=south west, 
            legend cell align={left},
            /pgfplots/table/ignore chars={|},
            mark options={solid},
        ]
        \addplot[mark=o,mark options={scale=1.5} , Paired-1 , semithick]  table[x=Eb/N0, y=FER] {data_EbNo/Polar/128_105/ORBGRAND_LW8256.txt};
        \addplot[mark=diamond,mark options={scale=2}  , Paired-4, semithick]  table[x=Eb/N0, y=FER] {data_EbNo/Polar/128_105/ORBGRAND_LW128.txt};
        \addplot[mark=star  , Paired-12, semithick]  table[x=Eb/N0, y=FER] {data_EbNo/Polar/128_105/ORBGRAND_LW96.txt};
        \addplot[mark=square  , Paired-3 , semithick]  table[x=Eb/N0, y=FER] {data_EbNo/Polar/128_105/ORBGRAND_LW96_HW8.txt}; 
        \addplot[mark=triangle, Paired-5 , semithick]  table[x=Eb/N0, y=FER] {data_EbNo/Polar/128_105/ORBGRAND_LW64.txt}; 
        \addplot[mark=pentagon, Paired-7 , semithick]  table[x=Eb/N0, y=FER] {data_EbNo/Polar/128_105/ORBGRAND_LW64_HW6.txt};
        \addplot[mark=*  , Paired-11, semithick, dashed]  table[x=Eb/N0, y=FER] {data_EbNo/Polar/128_105/GRANDAB.txt};
        \addplot[ Paired-9 , thick, dashed]  table[x=Eb/N0, y=FER] {data_EbNo/Polar/128_105/DSCF_N128_K105_C11_w2_Tmax50.txt}; 
        \addplot[mark=+,mark options={scale=1.2}  , Paired-11, semithick]  table[x=Eb/N0, y=FER] {data_EbNo/Polar/128_105/SCL_L32.txt};
    \end{semilogyaxis}
    \end{tikzpicture}
}
    \subfloat[BCH code (127, 106)]{
    \begin{tikzpicture}
    \begin{semilogyaxis}[
            footnotesize, width=\columnwidth, height=.67\columnwidth,    
            xmin=0, xmax=8, xtick={0,1,...,7},
            ymin=3e-8,  ymax=2,
            xlabel=$\frac{E_b}{N_0}$, ylabel=FER,  
            grid=both, grid style={gray!30},
            tick align=outside, tickpos=left, 
            legend pos=south west, 
            legend cell align={left},
            /pgfplots/table/ignore chars={|},
            mark options={solid},
        ]

        \addplot[mark=o, mark options={scale=1.5} , Paired-6 , semithick]  table[x=Eb/N0, y=FER] {data_EbNo/BCH/127_106/ORBGRAND_LW8128.txt};
        \addplot[mark=diamond, mark options={scale=2}  , Paired-8, semithick]  table[x=Eb/N0, y=FER] {data_EbNo/BCH/127_106/ORBGRAND_LW127_HW127.txt};
        \addplot[mark=star  , Paired-12, semithick]  table[x=Eb/N0, y=FER] {data_EbNo/BCH/127_106/ORBGRAND_LW96.txt};
        \addplot[mark=square  , Paired-3 , semithick]  table[x=Eb/N0, y=FER] {data_EbNo/BCH/127_106/ORBGRAND_LW96_HW8.txt}; 
        \addplot[mark=triangle, Paired-5 , semithick]  table[x=Eb/N0, y=FER] {data_EbNo/BCH/127_106/ORBGRAND_LW64.txt}; 
        \addplot[mark=pentagon, Paired-7 , semithick]  table[x=Eb/N0, y=FER] {data_EbNo/BCH/127_106/ORBGRAND_LW64_HW6.txt};
        \addplot[mark=*,Paired-13,semithick, dashed]  table[x=Eb/N0, y=FER] {data_EbNo/BCH/127_106/BCH_N127_106_BM.txt};
        \addplot[ mark=pentagon,Paired-13 , semithick]  table[x=Eb/N0, y=FER] {data_EbNo/BCH/127_106/BCH_N127_106_OSD.txt}; 

    \end{semilogyaxis}
    \end{tikzpicture}
}
\caption{Comparison of decoding performance of ORBGRAND decoding with different parameters ($LW_\text{max}$, $P$) for 5G CRC-aided Polar Code (128,105+11) and BCH code (127, 106).}
\label{fig:fer_polar_bch}
\end{figure*}

\section{VLSI Architecture for ORBGRAND}\label{sec:vlsi}

A VLSI architecture for GRANDAB (AB=3) was proposed in \cite{GRANDAB-VLSI} for $(n,k)$ linear block codes. Shift registers are used in \cite{GRANDAB-VLSI} to store syndrome of error patterns with a Hamming weight of $1$ (denoted as $\bm{s}_i = \bm{H}\cdot\mathds{1}_i^\top$, $i \in \llbracket 1\isep n \rrbracket$). Furthermore, \cite{GRANDAB-VLSI} employs the underlying code's linearity property to combine multiple $\bm{s}_i$ to generate syndrome of an error pattern with a Hamming weight of $l > 1$ ($\bm{s}_{1,2\ldots,l} = \bm{H}\cdot\mathds{1}_1^\top \oplus \bm{H}\cdot\mathds{1}_2^\top \ldots \oplus\bm{H}\cdot\mathds{1}_l^\top $). We refer the reader to \cite{GRANDAB-VLSI} for more details. The VLSI architecture for GRANDAB (AB=3) \cite{GRANDAB-VLSI} forms the basis for the proposed ORBGRAND architecture. The GRANDAB decoder \cite{GRANDAB-VLSI} can only generate test error patterns with Hamming weights $\leq3$. As a result, significant improvements are needed to cater to soft-inputs, to generate error patterns in increasing order of their logistic weight, and to consider larger Hamming weights as required by ORBGRAND.

\begin{figure}
  \centering
  \includegraphics[width=0.45\textwidth]{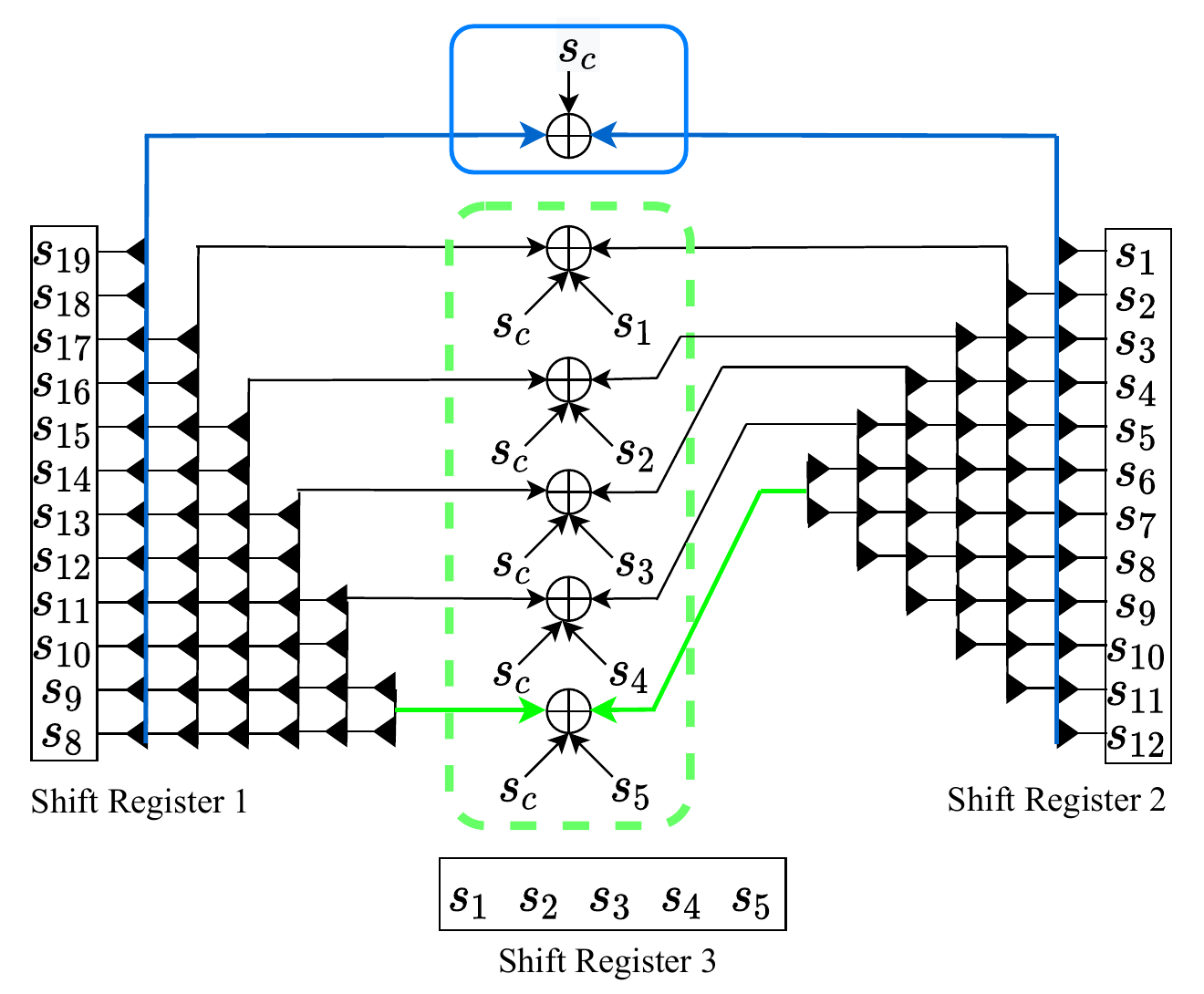}
  \caption{Example of the shift registers content and interconnection for logistic weight $m=20$ for checking error patterns of Hamming weights 2 and 3.}
  \label{fig:registers} 
\end{figure}
\begin{figure}[!ht]
  \centering
     \subfloat[Interconnections and the associated XOR gates for the first bus for checking error patterns of Hamming weight of 2 ($P=2$).\label{subfig-1:firstBus}]{%
     \centering
       \includegraphics[width=0.45\textwidth]{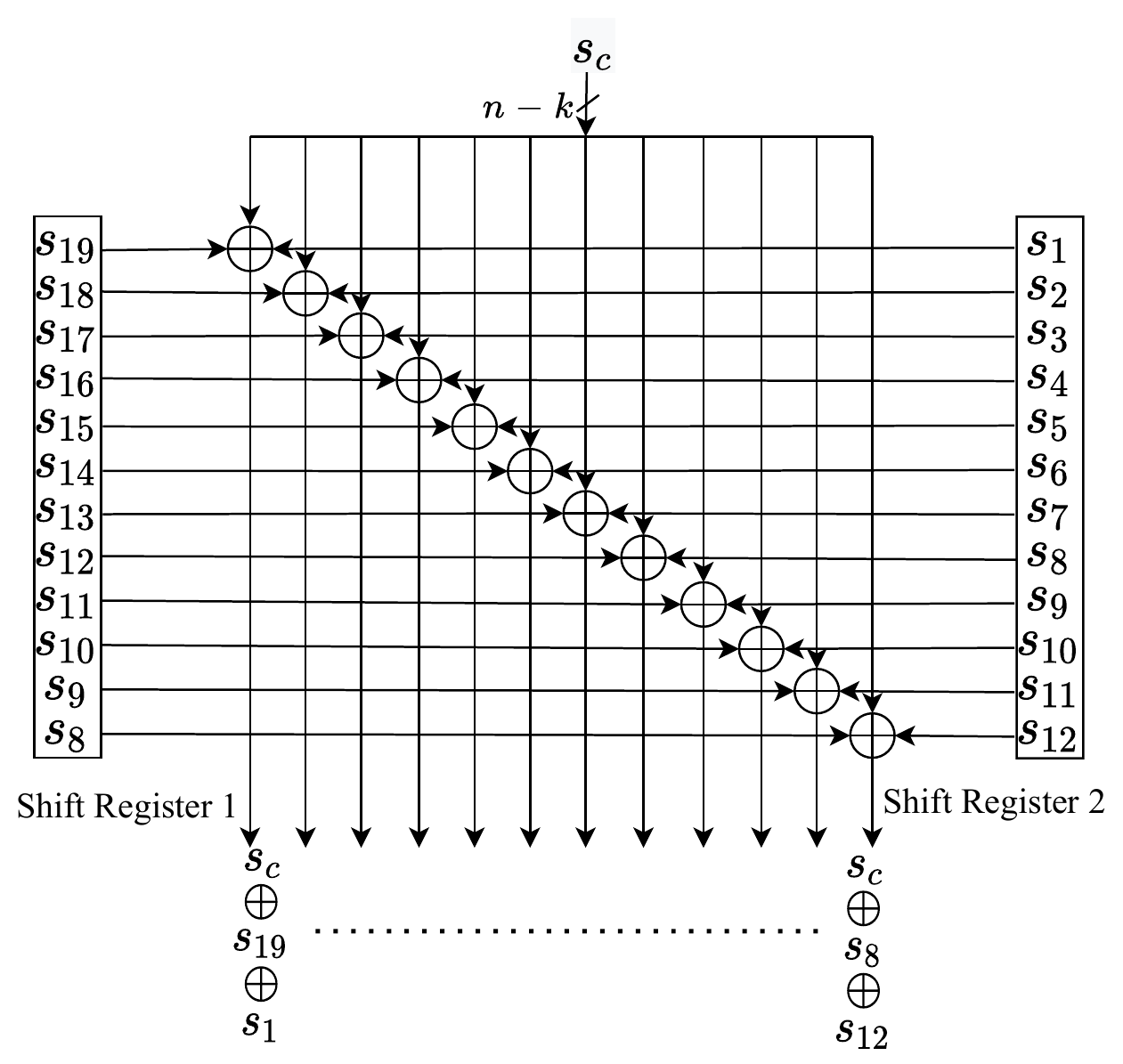}
     }
     \hfill
     \subfloat[Interconnections and the associated XOR gates for the last $(6^{th})$ bus for checking error patterns of Hamming weight of 3 ($P=3$).\label{subfig-2:lastBus}]{%
     \centering
       \includegraphics[width=0.45\textwidth]{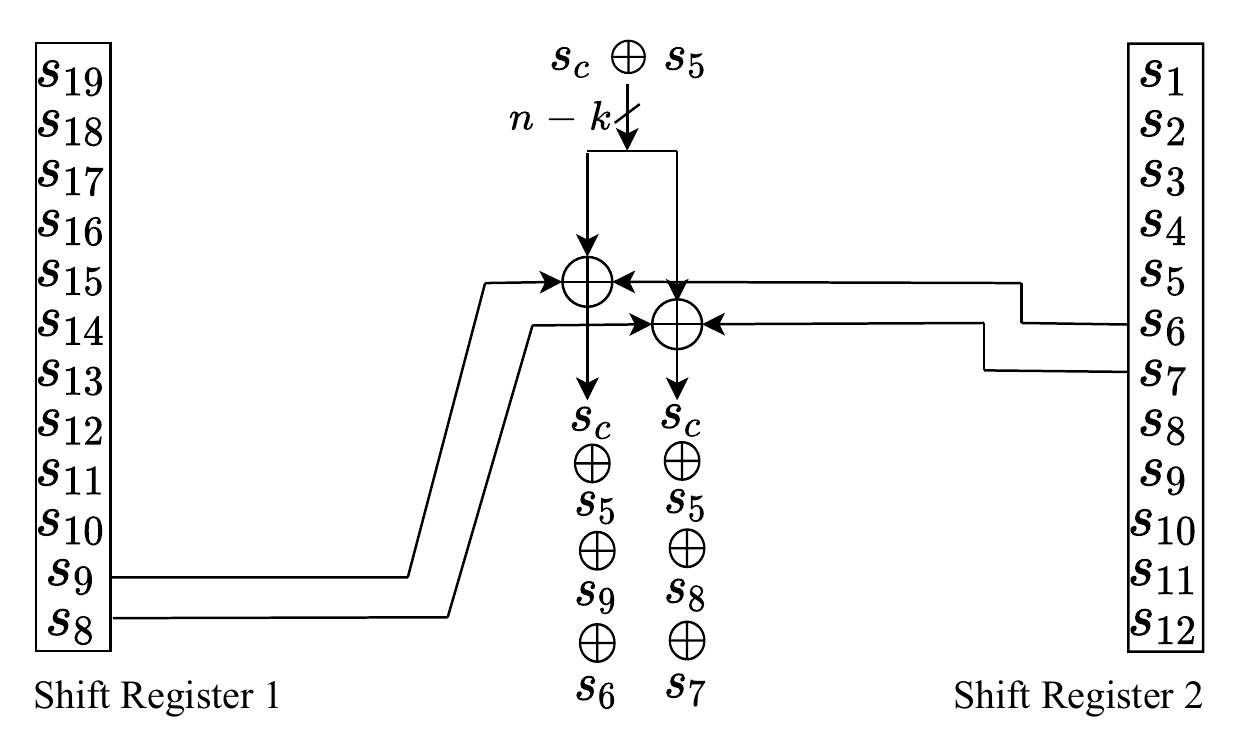}
     }
     \caption{Example of interconnections and the associated XOR gates for the first and last bus for logistic weight $m=20$.}
     \label{fig:cactus_1_a}
   \end{figure}

\subsection{Scheduling and details} \label{sec:Scheduling}
As explained in section \ref{sec:integerPartition}, ORBGRAND is based on generating test error patterns corresponding to integer partitions of a specific logistic weight $m$ ($\forall m \in [3,LW_\text{max}]$). Moreover, for each $m$, integer partitions are generated with size $P$ ($\forall P \in [2,P_\text{max}]$). We propose to generate these integer partitions in ascending order of their size ($P$). This modification does not impact the FER performance, however, it helps in designing a simpler hardware implementation.

In this section, we propose an arrangement and interconnection of shift registers and XOR gates to generate test error patterns corresponding to a specific logistic weight $m$. The shift registers store the syndromes that correspond to error patterns with a Hamming weight of $1$ ($\bm{s}_i$). To check for error patterns of hamming weight $P$, these syndromes are combined using an array of XOR gates.

\subsection{Generating test error patterns for $P\leq3$}

The size and number of shift registers used in \cite{GRANDAB-VLSI} have a direct impact on the Hamming weight of the error patterns that can be evaluated in parallel. For example, in \cite{GRANDAB-VLSI}, two $n\times(n-k)$ shift registers are used to evaluate $n$ test error patterns in parallel with a Hamming weight of 2. However, if more shift registers are added, the number of interconnections becomes a problem. As a result, for the proposed ORBGRAND architecture, we choose three shift registers that correspond to an integer partition of size 3 ($P=3$).

In the proposed ORBGRAND VLSI architecture, $\lambda_1$, $\lambda_2$, $\lambda_3$ $((\lambda_1, \lambda_2,\lambda_3) \vdash m)$ are mapped to first, second and third shift register respectively. The third shift register is a $\lambda_{3}^\text{max}\times(n-k)$ bit shift register, where $\lambda_{3}^\text{max}$ value is given by (\ref{eq:lambda_bound}) corresponding to $P = 3$. Whereas the first and second shift registers are each $2\times(\lambda_{3}^\text{max}+1)\times(n-k)$ bits in size. Since we have $\lambda_1 = m-\sum_{i=2}^{3}\lambda_i$, corresponding to $P = 3$, $\bm{s_{m-i}}$ is stored at the $i^{th}$ index of the first shift register, while for the second and third shift registers $\bm{s_i}$ is stored at the $i^{th}$ index. 

Fig. \ref{fig:registers} shows an example of the content and interconnection of three shift registers for logistic weight $m = 20$. The elements of the three shift registers are syndromes ($\bm{s_i}$) of the error pattern with Hamming weight of $1$. These syndromes ($\bm{s_i}$)  of the error pattern with Hamming weight of $1$ are combined using an array of $(n-k)$-wide XOR gates to check for error patterns with Hamming weights $2$ and $3$.

A collection of these connections is defined as a \textit{bus}. Since there are numerous connections and XOR gates involved, we used a single XOR gate and a single \textit{bus} symbol to illustrate these interconnections in Fig. \ref{fig:registers}. As seen in Fig. \ref{fig:registers}, there are $6$ buses ($\lambda_{3}^\text{max}+1$, where $\lambda_{3}^\text{max}=5$ for $P=3$ (\ref{eq:lambda_bound})) for $m = 20$. The first bus (highlighted by solid rectangle) is used to check error patterns with Hamming weight of $2$, and the remaining buses (highlighted by the dashed rectangle in Fig. \ref{fig:registers}) are used to check for error patterns of Hamming weight 3.

To check the error patterns corresponding to a Hamming weight 2 ($P=2$), the first bus (highlighted by solid rectangle) is used to combine all the elements of shift register 1 with all the elements of the shift register 2 using an array of XOR gates. These results are again combined with the syndromes of the received vector ($\bm{s_c}$) to check for the error patterns with Hamming weight of $2$. The detailed interconnections and the associated XOR gates for the first bus are shown in Fig. \ref{subfig-1:firstBus}. 

Similarly, to check the error patterns corresponding to a Hamming weight of 3 ($P=3$), the selected elements of the shift register 1 and 2 are again combined with $\bm{s_c}$, but also with the elements of the shift register 3. We use a single bus and a single XOR gate to illustrate these interconnections, which are depicted in Fig. \ref{fig:registers} by the dashed rectangle.  The detailed interconnections and the associated XOR gates for the last ($6^{th}$) bus are shown in Fig. \ref{subfig-2:lastBus}. 

\begin{figure}
  \centering
  \includegraphics[width=0.5\textwidth]{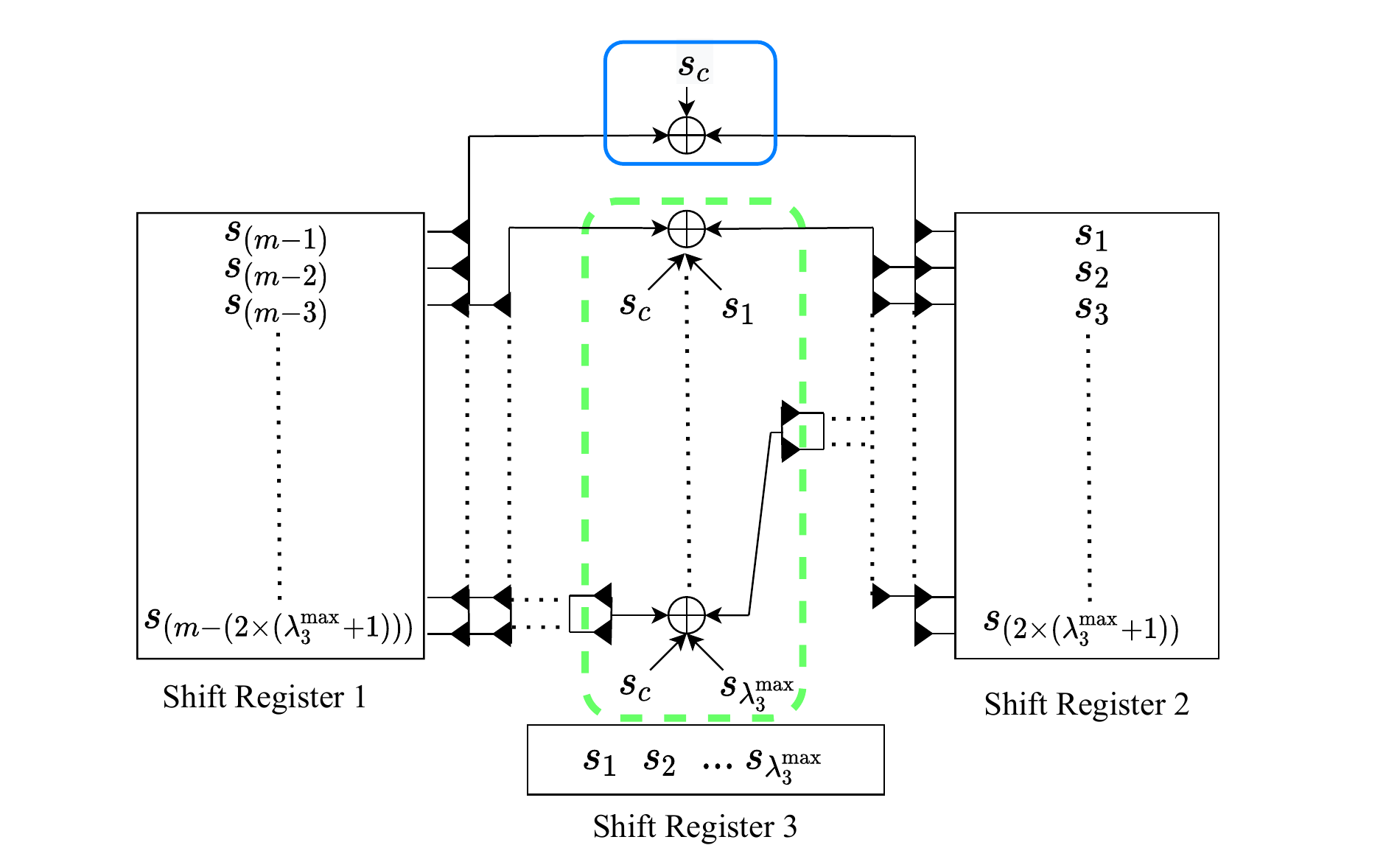}
  \caption{Shift registers contents for checking error patterns corresponding to a Hamming weight of 2 and 3 for any logistic weight $m$.}
  \label{fig:registersPN1} 
\end{figure}
\begin{figure}
  \centering
  \includegraphics[width=0.5\textwidth]{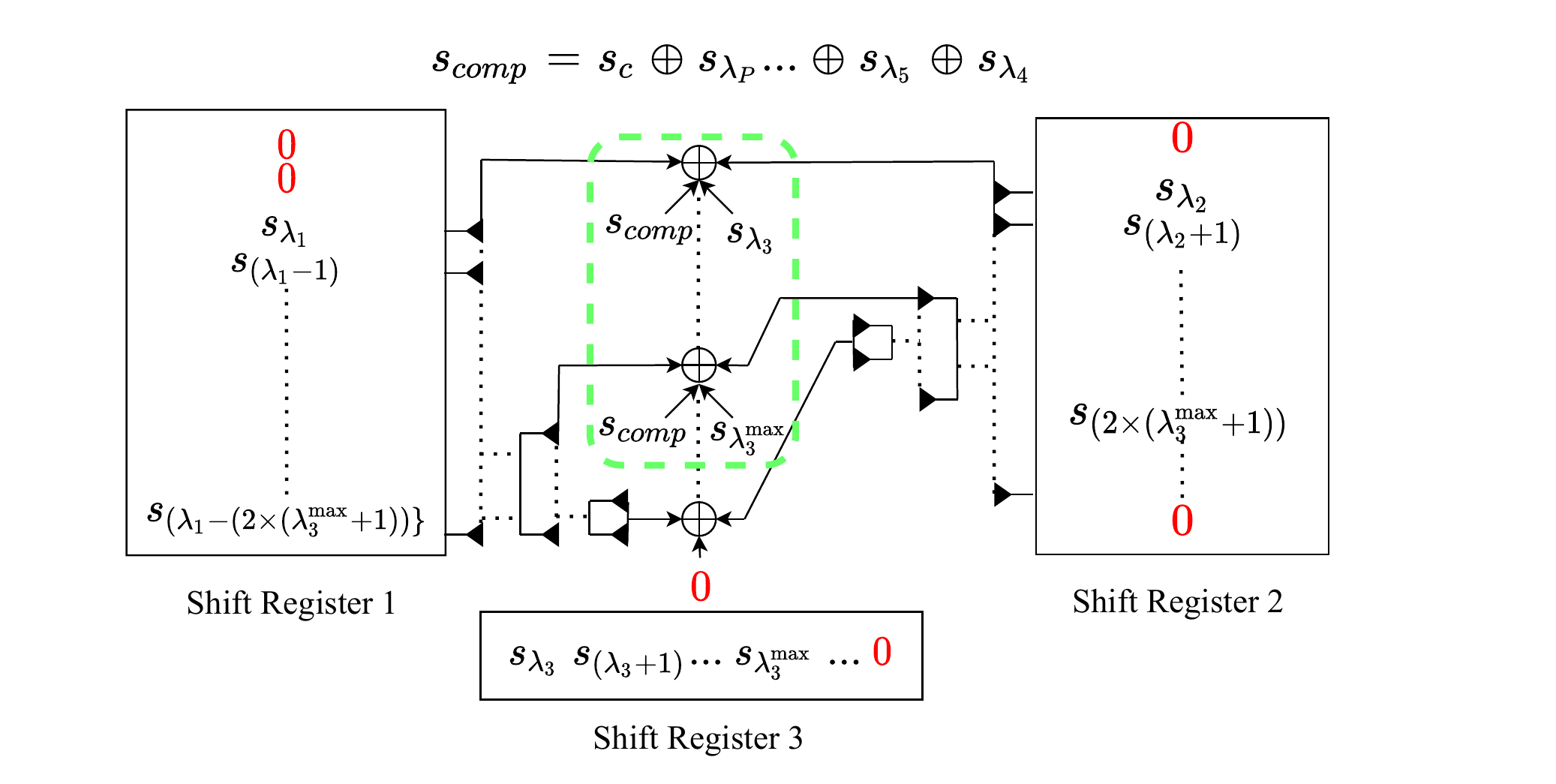}
  \caption{Shift registers contents for checking error patterns corresponding to $P>3$ for any logistic weight $m$ ($\lambda_1 = m-\sum_{i=2}^{P}\lambda_i$).}
  \label{fig:registersPN2} 
\end{figure}
\begin{figure}[!ht]
  \centering
     \subfloat[Shift registers contents for checking test error patterns corresponding to $P=4$ at time step 1.\label{subfig-1:registersP41}]{%
     \centering
       \includegraphics[width=0.45\textwidth]{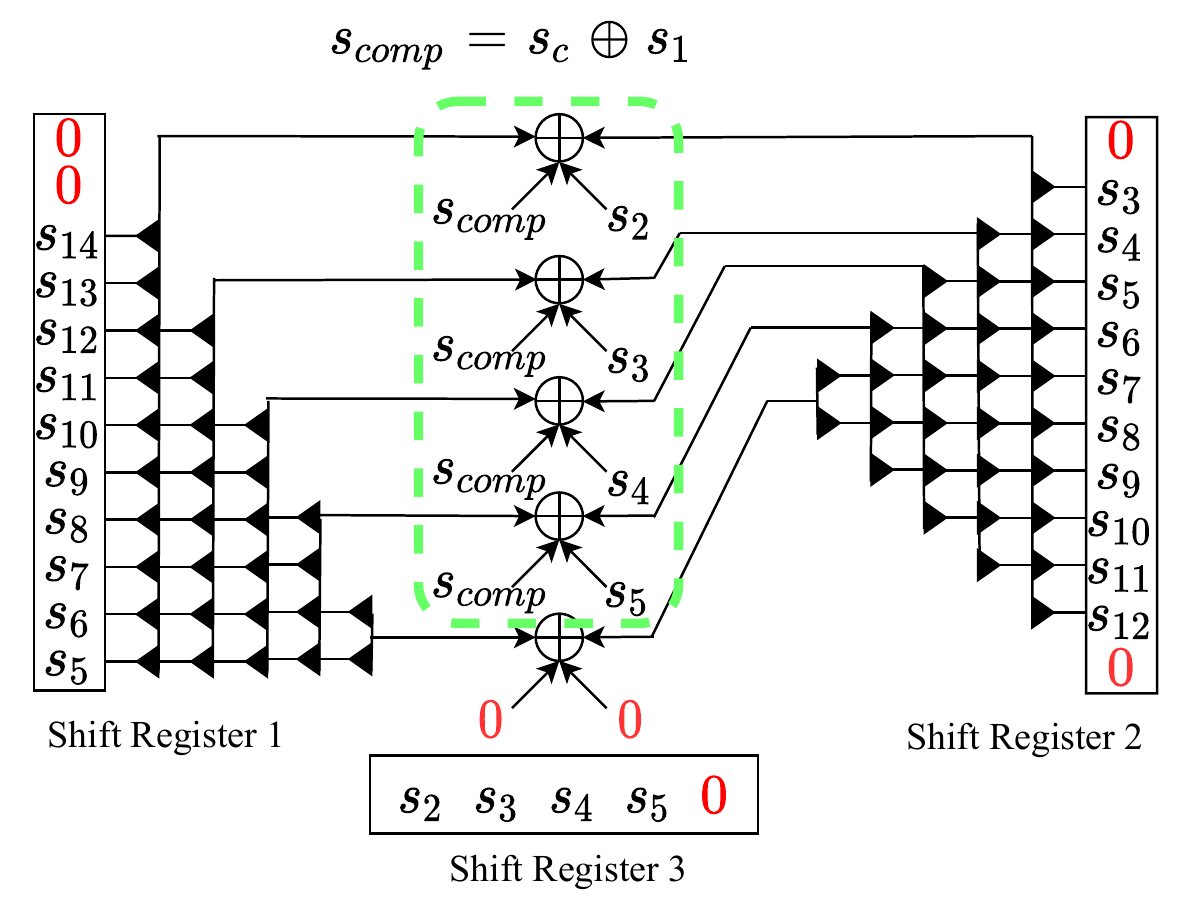}
     }
     \hfill
     \subfloat[Shift registers contents for checking test error patterns corresponding to $P=4$ at time step 2.\label{subfig-2:registersP42}]{%
     \centering
       \includegraphics[width=0.45\textwidth]{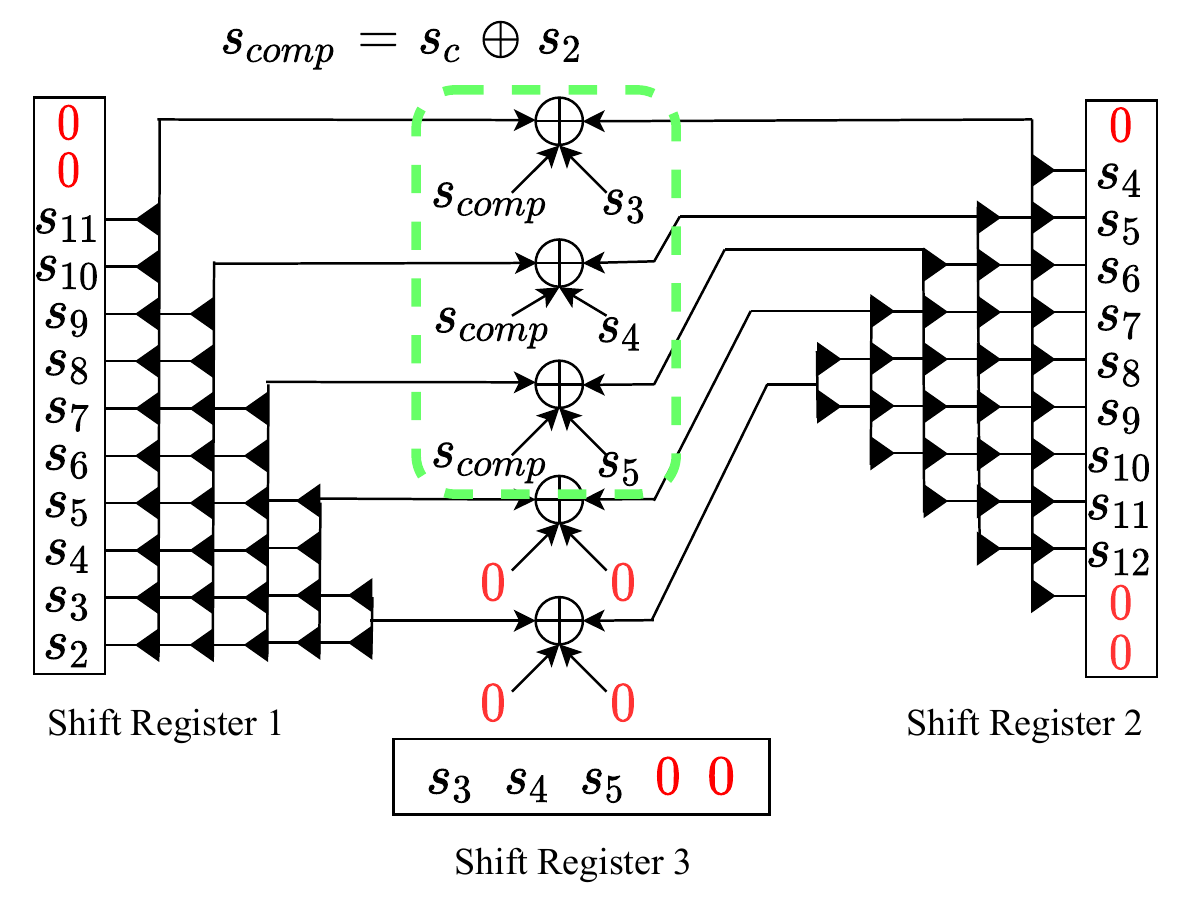}
     }
     \hfill
     \subfloat[Shift registers contents for checking test error patterns corresponding to $P=4$ at time step 3.\label{subfig-3:registersP43}]{%
     \centering
       \includegraphics[width=0.45\textwidth]{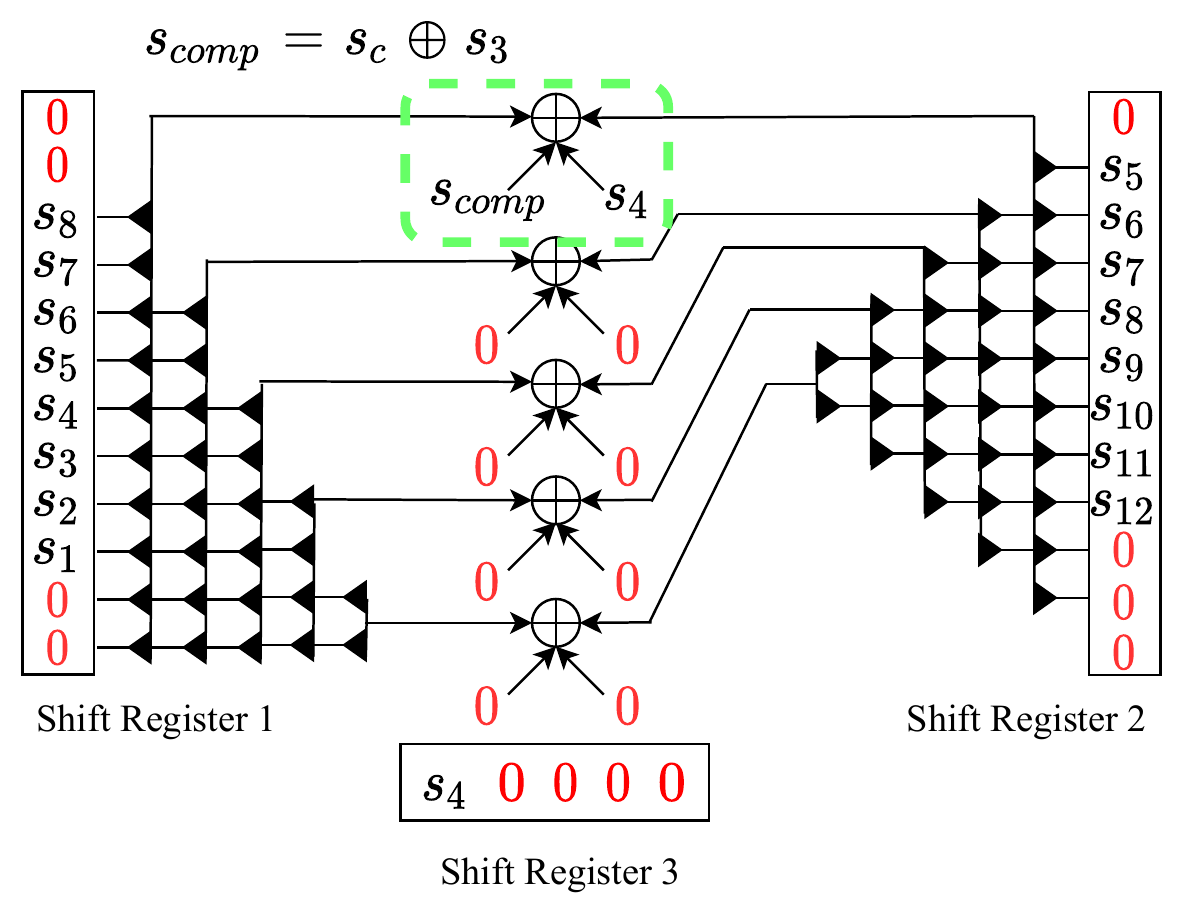}
     }
     \caption{Shift registers contents for checking test error patterns corresponding to to $P=4$ ($m=20$).}
     \label{fig:registersP4}
   \end{figure}
\begin{figure}[!ht]
  \centering
     \subfloat[Shift registers contents for checking test error patterns corresponding to $P=5$ at time step 1.\label{subfig-1:registersP51}]{%
     \centering
       \includegraphics[width=0.45\textwidth]{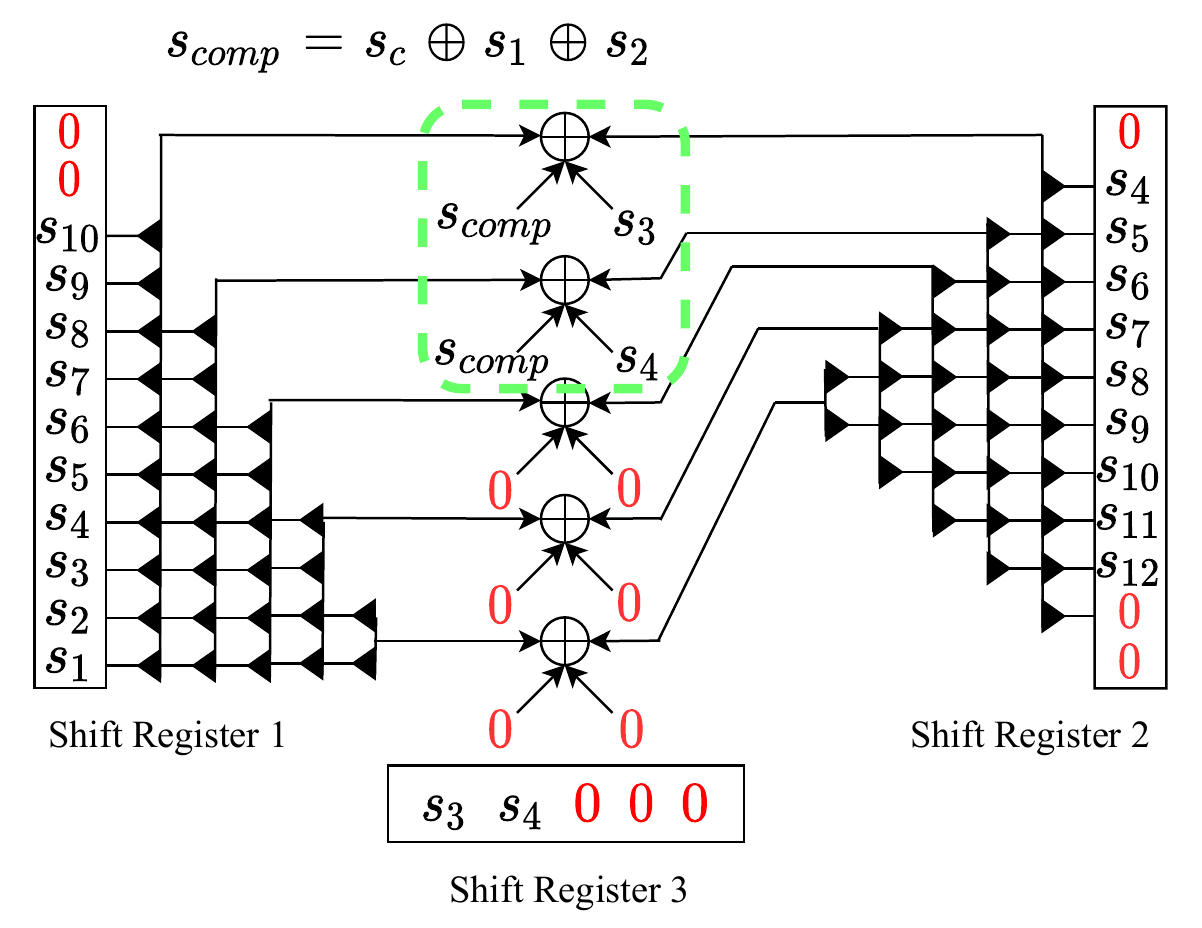}
     }
     \hfill
     \subfloat[Shift registers contents for checking test error patterns corresponding to $P=5$ at time step 2.\label{subfig-1:registersP52}]{%
     \centering
       \includegraphics[width=0.45\textwidth]{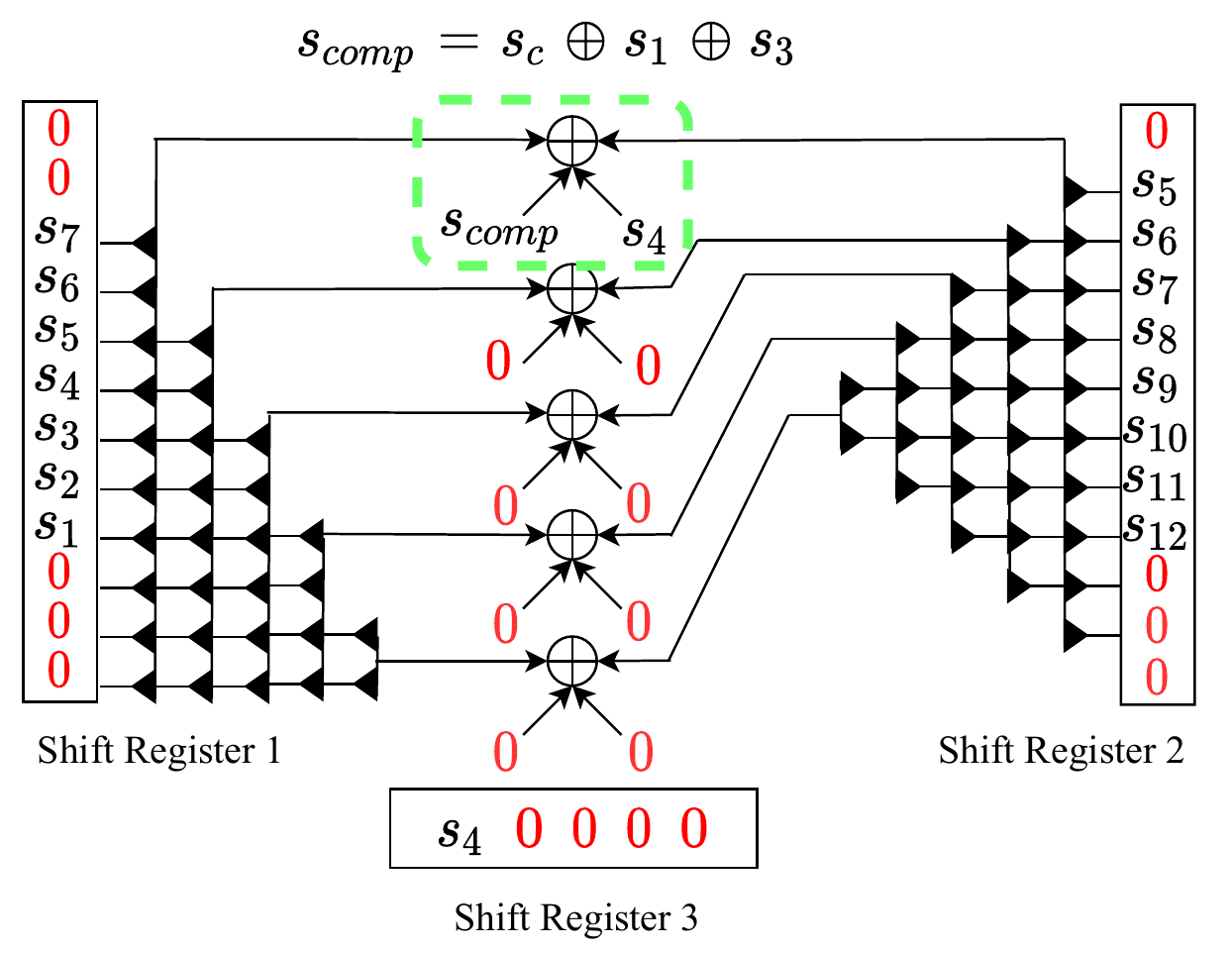}
     }
     \hfill
     \subfloat[Shift registers contents for checking test error patterns corresponding to $P=5$ at time step 3.\label{subfig-2:registersP53}]{%
     \centering
       \includegraphics[width=0.45\textwidth]{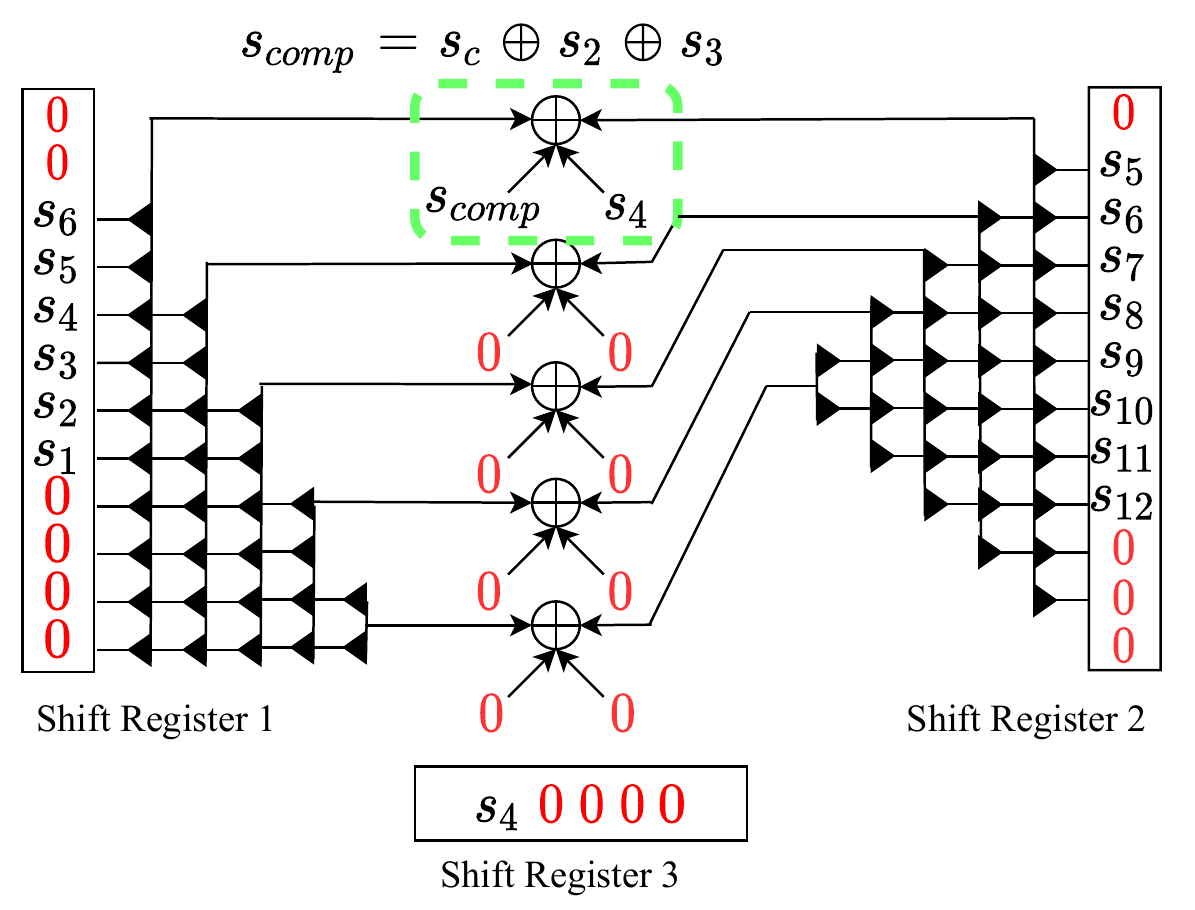}
     }
     \caption{Shift registers contents for checking test error patterns corresponding to $P=5$ ($m=20$).}
     \label{fig:registersP5}
   \end{figure}
Due to the described arrangement and interconnection of the shift registers and XOR gates, all the error patterns corresponding to an integer partition of sizes 2 and 3 for a specific logistic weight $m$ are checked in one time-step. In general, to check the error patterns corresponding to an integer partition of sizes 2 and 3 for any logistic weight $m$, the content and the interconnection of the three shift registers are depicted in Fig. \ref{fig:registersPN1}.

\subsection{Generating test error patterns for $P>3$}

To check all the test error patterns corresponding to integer partitions of sizes $P>3$, a controller is used in conjunction with the shift registers. The controller combines $P_\text{max} - 3$ syndromes together with the syndromes of the received vector, noted $\bm{s}_{comp}$. Hence, when $\bm{s}_{comp}$ is fixed, only one time-step is required to generate all possible combinations of $\{\lambda_1, \lambda_2, \lambda_3\}$ using the shift registers with adequately chosen shift values. 

The content and the interconnection of the three shift registers, which are used to check the test error patterns corresponding to integer partitions of sizes $P>3$, are depicted in Fig. \ref{fig:registersPN2}. Since the first bus is only used to check error patterns with Hamming weight of $2$ ($P=2$), it is disabled for $P>3$ and not shown in Fig. \ref{fig:registersPN2}. A $\bm{0}$ corresponds to a disabled connection, which means the respective elements of the bus, do not take part in the final computations. Fig. \ref{fig:registersP4} illustrates testing error patterns corresponding to $P=4$. At each time step, the controller outputs $\bm{s}_{comp}=\bm{s}_c\oplus\bm{s}_{\lambda_4}$ ($\lambda_4\in [1, \lambda_4^{max}]$) and $\{\lambda_1, \lambda_2, \lambda_3\}$ are computed and mapped to their corresponding shift registers.

At the first time step, having received $\bm{s}_{comp}=\bm{s}_c\oplus\bm{s}_1$ ($\lambda_4=1$) as an output from the controller, $\lambda_3$  ($\lambda_3\in [2, \lambda_{3,\lambda_4}^{max}]$ where $\lambda_{3,\lambda_4}^{max}=5$ with $\lambda_4=1$ $(\ref{eq:lambda_bound})$) is computed and mapped to the third shift register. Similarly, $\lambda_2$ ($\lambda_2\in [\lambda_3+1, \lambda_{(2\times(\lambda_3^\text{max}+1)))}]$) and $\lambda_1$ ($\lambda_1 = m-\sum_{i=2}^{4}\lambda_i$) are computed and mapped to their corresponding shift registers. The test error patterns with $\lambda_4=1$ are checked as shown in Fig. \ref{subfig-1:registersP41}.
 
At the next time step, the controller outputs $\bm{s}_{comp}=\bm{s}_c\oplus\bm{s}_2$ ($\lambda_4=2$) and $\lambda_3$  ($\lambda_3\in [3, \lambda_{3,\lambda_4}^{max}]$ where $\lambda_{3,\lambda_4}^{max}=5$ with $\lambda_4=2$ $(\ref{eq:lambda_bound})$) is computed. Shift register 2 is shifted up by 1 position and shift register 1 outputs $\lambda_1$ $(\lambda_1 = m-\sum_{i=2}^{4}\lambda_i)$ as shown in Fig. \ref{subfig-2:registersP42}. Hence, the test error patterns with $\lambda_4=2$ are checked in the second time-step. Similarly, at third time step, the controller outputs $\bm{s}_{comp}=\bm{s}_c\oplus\bm{s}_3$, ($\lambda_4=3$) $\lambda_3$ ($\lambda_3\in [4, \lambda_{3,\lambda_4}^{max}]$ where $\lambda_{3,\lambda_4}^{max}=4$ with $\lambda_4=3$ $(\ref{eq:lambda_bound})$) is computed as shown in Fig. \ref{subfig-3:registersP43}. Therefore, a total of 3 time steps ($\lambda_4^{max}=3$, Eq. $\ref{eq:lambda_bound}$),  are required to check for error patterns corresponding to $P=4$ and  $m = 20$. 

Fig. \ref{fig:registersP5} depicts the use of shift registers to check the error patterns corresponding to $P=5$ and  $m = 20$. At each time step, the controller outputs $\bm{s}_{comp}=\bm{s}_c\oplus\bm{s}_{\lambda_5}\oplus\bm{s}_{\lambda_4}$. For each value of $\lambda_5$ ($\lambda_5\in [1, \lambda_5^{max}]$), $\lambda_4$ ($\lambda_4\in [\lambda_5+1, \lambda_{4,\lambda_5}^{max}]$) is computed. Similarly, for each value of $\lambda_4$, $\lambda_3$ ($\lambda_3\in [\lambda_4+1, \lambda_{3,\lambda_5,\lambda_4}^{max}]$), $\lambda_2$ ($\lambda_2\in [\lambda_3+1, \lambda_{(2\times(\lambda_3^\text{max}+1))}]$) and $\lambda_1$ ($\lambda_1 = m-\sum_{i=2}^{5}\lambda_i$) are computed and mapped to their corresponding shift registers. Hence, a total of $3$ time steps ($\sum_{\lambda_5=1}^{\lambda_{5}^\text{max}}\left(\sum_{\lambda_4={\lambda_5+1}}^{\lambda_{4,\lambda_5}^\text{max}}\left(1\right)\right)$, where $\lambda_{5}^\text{max}=2$, $\lambda_{4,\lambda_5=1}^\text{max}=3$ and $\lambda_{4,\lambda_5=2}^\text{max}=3$ $(\ref{eq:lambda_bound})$) are required to check for error patterns corresponding to $P=5$ and  $m = 20$ as shown in Fig. \ref{fig:registersP5} 
In general, the number of time steps required to generate all integer partitions of size $P>3$ for a specific logistic weight ($LW$) is given by:
\begin{equation}
\label{eq:nb_steps}
\sum_{\lambda_P=1}^{\lambda_{P}^\text{max}}\left(\sum_{\lambda_{P-1}=\lambda_P+1}^{\lambda_{{P-1},\lambda_P}^\text{max}}\left(\ldots\sum_{\lambda_{4}=\lambda_5+1}^{\lambda_{4,\lambda_5,\ldots,\lambda_P}^\text{max}}\left(1\right)\right)\right).  
\end{equation}
\begin{figure*}
\centering
  \includegraphics[width=0.85\linewidth]{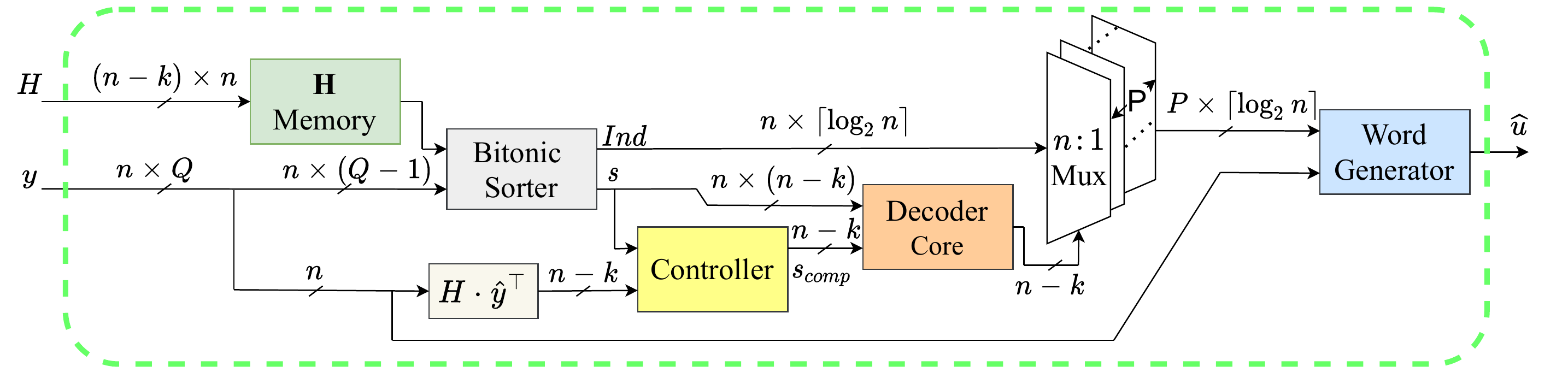}
  \caption{Proposed VLSI Architecture for ORBGRAND.}
  \label{fig:arch}
\end{figure*}
\subsection{Proposed VLSI architecture}
Figure~\ref{fig:arch} depicts the proposed VLSI architecture for ORBGRAND which can be used to decode any linear block code of length $n$. For clarity, the control and clock signals are not shown. To support different codes and rates, any $\bm{H}$  matrix can be loaded into the \textit{H memory} of size $(n-k)\times n\text{-bit}$ at any time. The hard decided vector $\hat{\bm{y}}$ is subjected to a syndrome check (\ref{eq:constraint}) in the first phase of decoding. Decoding is assumed to be successful if the syndrome is verified. Otherwise, the LLRs values are sorted in ascending order of their absolute value $\vert\bm{y}\vert$.

As depicted in Fig.~\ref{fig:arch}, the sorted syndromes of error patterns with Hamming weight of $1$ ($\bm{s_i}$) are passed to the \textit{decoder core}, while the indices of the sorted LLRs are forwarded to the multiplexers for later use by the \textit{word generator} module. Following the sorting process, all syndromes of error patterns with Hamming weight of $1$ ($\bm{s_i}$) are tested for codebook membership (\ref{eq:constraint}) in a single time-step. Following that, error patterns are tested for codebook membership in ascending logistic weight ($LW$) order as explained in section \ref{sec:ORBGRANDdecoding}. 

The test error pattern syndromes corresponding to integer partitions of a given logistic weight $m$ are generated using the shift register and XOR gate arrangement proposed in section \ref{sec:Scheduling}. The rows of shift registers are combined with the controller's  output ($\bm{s}_{comp}$), and the resulting test syndromes are NOR-reduced and fed to a 2D priority encoder. Each NOR-reduce output is 1 if and only if all of the bits of the syndromes computed by  (\ref{eq:constraint}) are 0. If any of the tested syndrome combinations satisfy the parity check constraint (NOR-reduced output is 1), the 2D priority encoder is used in conjunction with the \textit{controller} module to forward the respective indices to the word generator module, where $P$ multiplexers are used to convert the sorted index values to their appropriate bit-flip locations.

\begin{figure}[!t]
\centering
  \begin{tikzpicture}
    \begin{groupplot}[group style={group name=lat_tp, group size= 2 by 1, horizontal sep=10pt, vertical sep=10pt}, 
                      footnotesize,
                      height=.6\columnwidth,  width=.55\columnwidth,
                      xlabel=LW,
                      xmin=32, xmax=128, xtick={20,40,...,140},
                      ymode=log,
                      tick align=inside, 
                      grid=both, grid style={gray!30},
             ]

      \nextgroupplot[ylabel= W.C. Latency (cycles), ytick pos=left, y label style={at={(axis description cs:-0.15,.5)},anchor=south}, ymax = 1e6]
        \addplot[ Paired-9, semithick]  table[x=LW, y=Lat] {data_EbNo/Polar/latency/Latency_HW_4.txt};\label{gp:plot1}
        \addplot[Paired-7, semithick]  table[x=LW, y=Lat] {data_EbNo/Polar/latency/Latency_HW_5.txt};\label{gp:plot2}
        \addplot[Paired-5, semithick]  table[x=LW, y=Lat] {data_EbNo/Polar/latency/Latency_HW_6.txt};\label{gp:plot3}
        \addplot[ Paired-3, semithick]  table[x=LW, y=Lat]  {data_EbNo/Polar/latency/Latency_HW_7.txt};\label{gp:plot4}
        \addplot[Paired-1, semithick]  table[x=LW, y=Lat]  {data_EbNo/Polar/latency/Latency_HW_8.txt};\label{gp:plot5}
        \coordinate (top) at (rel axis cs:0,1);

      \nextgroupplot[ylabel=W.C. T/P (Mbps), ytick pos=right,y label style={at={(axis description cs:1.33,.5)},anchor=south}, ,ymin=5e-2, ymax = 1e3]
        \addplot[ Paired-9, semithick]  table[x=LW, y=TP] {data_EbNo/Polar/latency/Latency_HW_4.txt};
        \addplot[ Paired-7, semithick]  table[x=LW, y=TP] {data_EbNo/Polar/latency/Latency_HW_5.txt};
        \addplot[ Paired-5, semithick]  table[x=LW, y=TP] {data_EbNo/Polar/latency/Latency_HW_6.txt};
        \addplot[ Paired-3, semithick]  table[x=LW, y=TP]  {data_EbNo/Polar/latency/Latency_HW_7.txt};
        \addplot[ Paired-1, semithick]  table[x=LW, y=TP]  {data_EbNo/Polar/latency/Latency_HW_8.txt};

        \coordinate (bot) at (rel axis cs:1,0);
    \end{groupplot}
    \node[below = 1cm of lat_tp c1r1.south] {(a) : W.C. Latency};
    \node[below = 1cm of lat_tp c2r1.south] {(b) : W.C. Info. Throughput};
    \path (top|-current bounding box.north) -- coordinate(legendpos) (bot|-current bounding box.north);
    \matrix[
        matrix of nodes,
        anchor=south,
        draw,
        inner sep=0.2em,
        draw
      ]at(legendpos)
      {
        \ref{gp:plot1}& \footnotesize P = 4 &[5pt]
        \ref{gp:plot2}& \footnotesize P = 5  
        \ref{gp:plot3}& \footnotesize P = 6  &[5pt]
        \ref{gp:plot4}& \footnotesize P = 7 
        \ref{gp:plot5}& \footnotesize P = 8 \\};
  \end{tikzpicture}
  \caption{\label{fig:lat_tp} Worst-Case (W.C.) latency and W.C. information throughput for the proposed ORBGRAND architecture with various parameters ($LW$, $P$).}
\end{figure}
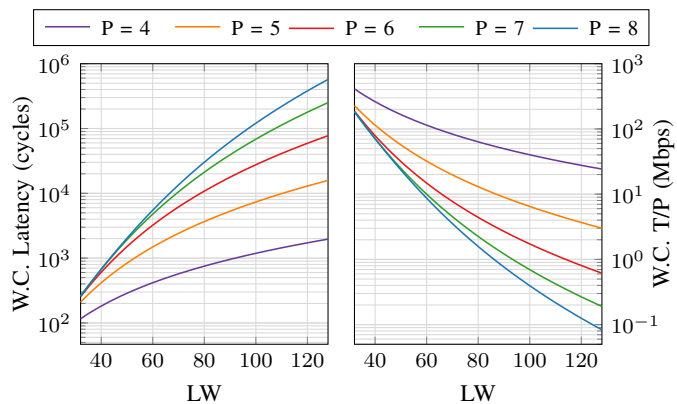

\begin{table*}[t]
\centering
\caption{\label{table:tableORBGRAND_128}TSMC 65 nm CMOS Synthesis Comparison for ORBGRAND with GRANDAB and DSCF for $n=128$.}
\begin{adjustbox}{max width=\textwidth}
\begin{tabular}{lrrrrrr}
\toprule
                             & GRANDAB \cite{GRANDAB-VLSI}  & \multicolumn{4}{c}{ORBGRAND}   & DSCF \cite{Ercan_2020}        \\
                             \cmidrule(l){2-2}\cmidrule(l){3-6}\cmidrule(l){7-7}
Parameters                   & AB=3    & LW$\leq$64, P$\leq$6 & LW$\leq$96, P$\leq$8  & LW$\leq$96, P$\leq$8, $S=2$ & LW$\leq$96, P$\leq$8, $S=4$  & $\omega = 2$, $T_\text{max}=50$ \\  
Technology (nm)              & 65      & 65                   & 65             & 65     & 65     & 65                   \\
Supply (V)                   & 0.9     & 0.9                  & 0.9            & 0.9    & 0.9    & 0.9                   \\
Max. Frequency (MHz)         & 500     & 454                  & 454            & 454    & 454    & 426                   \\
Area (mm\textsuperscript{2}) & 0.25    & 1.82                 & 2.25           & 2.08   & 1.85   & 0.22                  \\
W.C. Latency ($\mu$s)        & 8.196   & 9.30                 & 205.76         & 205.76 & 205.76 & 6.103                 \\
Avg. Latency (ns)            & 2       & 2.47                 & 2.47           & 2.47   & 2.47   & 122                   \\
W.C. T/P (Mbps)$^a$              & 12.8    & 11.3                 & 0.51          & 0.51  & 0.51  & 17.2                  \\
Avg. T/P (Gbps)$^a$              & 52.5    & 42.5                 & 42.5           & 42.5   & 42.5   & 0.86                  \\
Power (mW)                   & 46      & 104.3                & 133            & 131.3  & 130    & 68.51                 \\
Energy per Bit (pJ/bit)$^b$                & 0.87    & 2.45                 & 3.13            & 3.09   & 3.0    & 79.6                   \\
Area Efficiency (Gbps/mm\textsuperscript{2})$^c$ & 210   & 23.3  & 18.9            & 20.4   & 23    & 3.9                      \\
Code compatible              & Yes     & Yes                  & Yes            & Yes    & Yes   & No                       \\
\bottomrule
\multicolumn{7}{l}{\footnotesize \textcolor{black} {$^a$ $\text{Information Throughput (Gbps)}=\frac{{k}}{\text{Decoding Latency (ns)}}$}, \textcolor{black} {$^b$ $\text{Energy per Bit (pJ/bit)}=\frac{\text{Power (mW)}}{\text{Avg. Throughput (Gbps)}}$}, \textcolor{black} {$^c$ $\text{Area Efficiency (Gbps/mm\textsuperscript{2})}=\frac{\text{Avg. Throughput (Gbps)}}{\text{Area (mm\textsuperscript{2})}}$} } \\
\end{tabular}
\end{adjustbox}
\end{table*}

\begin{table}[t]
\centering
\caption{\label{table:tableSorter} Displacement of LLR elements $\bm{y}_i$ ($\forall i \in [1,n]$) from their correct locations with segmented sorter}
\begin{adjustbox}{max width=\columnwidth}
\begin{tabular}{ccccc}
\toprule
               & \multicolumn{4}{c}{\textit{\# of segments for bitonic sorter} } \\
\textit{Displacement}    & $S=2$       & $S=4$       & $S=8$       & $S=16$      \\
$=0$                      & $10.31\%$   & $5.98\%$    & $3.87\%$    & $2.59\%$    \\
$\leq1$                   & $29.40\%$   & $17.42\%$   & $11.40\%$   & $7.67\%$   \\
$\leq2$                   & $45.50\%$   & $28.18\%$   & $18.67\%$   & $12.65\%$   \\
$\leq3$                   & $58.76\%$   & $38.05\%$   & $25.67\%$   & $17.50\%$   \\
$\leq5$                   & $77.84\%$   & $54.62\%$   & $38.65\%$   & $26.85\%$   \\
$\leq10$                  & $96.89\%$   & $81.64\%$   & $63.82\%$   & $47.68\%$   \\
$\leq20$                  & $99.99\%$   & $98.34\%$   & $90.09\%$   & $75.95\%$   \\
$\leq30$                  & $100\%$     & $99.94\%$   & $98.10\%$   & $90.58\%$  \\
\bottomrule
\end{tabular}
\end{adjustbox}
\end{table}

\begin{figure}
  \centering
  \includegraphics[width=0.5\textwidth]{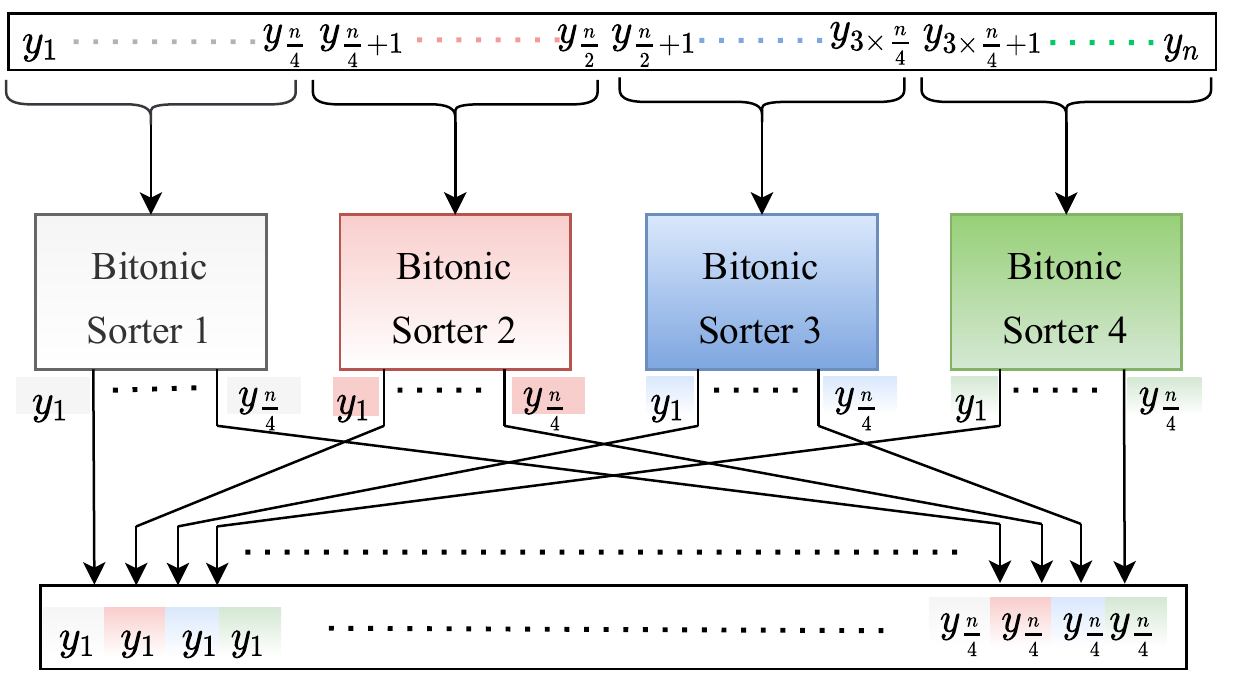}
  \caption{Proposed segmented sorter ($S=4$) for ORBGRAND.}
  \label{fig:segSorter} 
\end{figure}

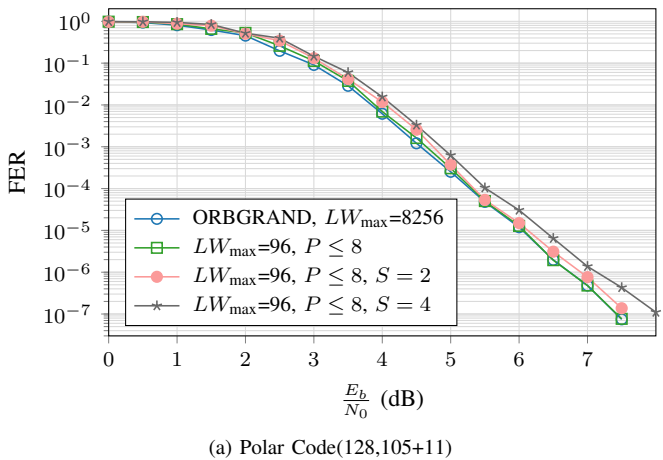
\begin{figure}[!t]
    \centering
    \captionsetup[subfigure]{oneside}
    \subfloat[Polar Code(128,105+11)]{
    \begin{tikzpicture}
    \begin{semilogyaxis}[
            footnotesize, width=\columnwidth, height=.67\columnwidth,    
            xmin=0, xmax=8, xtick={0,1,...,7},
            ymin=3e-8,  ymax=2,
            xlabel=$\frac{E_b}{N_0}$ (dB), ylabel=FER,  
            grid=both, grid style={gray!30},
            tick align=outside, tickpos=left, 
            legend pos=south west, 
            legend cell align={left},
            /pgfplots/table/ignore chars={|},
            mark options={solid},
        ]
        \addplot[mark=o       , Paired-1 , semithick]  table[x=Eb/N0, y=FER] {data_EbNo/Polar/128_105/ORBGRAND_LW8256.txt};
        \addplot[mark=square  , Paired-3 , semithick]  table[x=Eb/N0, y=FER] {data_EbNo/Polar/128_105/ORBGRAND_LW96_HW8.txt}; 
        \addplot[mark=*  , Paired-6, semithick]  table[x=Eb/N0, y=FER] {data_EbNo/Polar/128_105/ORBGRAND_LW96_HW8_S_2.txt};
        \addplot[mark=star  , Paired-12, semithick]  table[x=Eb/N0, y=FER] {data_EbNo/Polar/128_105/ORBGRAND_LW96_HW8_S_4.txt};
        \legend{{} {\footnotesize{{ORBGRAND}, $LW_\text{max}$=8256}},
        {} {\footnotesize{$LW_\text{max}$=96, $P \leq 8$}},
        {} {\footnotesize{$LW_\text{max}$=96, $P \leq 8$, $S = 2$}}, 
        {} {\footnotesize{$LW_\text{max}$=96, $P \leq 8$, $S = 4$}}, }
    \end{semilogyaxis}
    \end{tikzpicture}
}
\caption{Comparison of decoding performance of ORBGRAND decoding with different parameters ($LW_\text{max}$, $P$, $S$) for Polar Code(128,105+11)}
\label{fig:fer_polar_seg_sort}
\end{figure}

{\section{ORBGRAND Design Exploration}\label{sec:DSE}

In this section, we present the VLSI implementation results for ORBGRAND ($LW$,$P$). Initially, implementation results for baseline ORBGRAND, which can support parameters $LW \leq 64$ and $P \leq 6$, are presented. Please note that a subset of these implementation results were previously presented in \cite{ORBGRAND-VLSI}. In this work, baseline implementation results are supplemented with power consumption, area, and energy efficiency results. Furthermore, the proposed ORBGRAND VLSI architecture is extended to support the parameters $LW \leq 96$ and $P \leq 8$, and a comprehensive analysis of worst-case latency and worst-case throughput for selecting different parameters ($LW$, $P$) for the proposed ORBGRAND hardware is presented. Finally, the sorter module for the proposed ORBGRAND is segmented into multiple partitions to reduce the area overhead. The effect of the number of partitions on ORBGRAND decoding performance as well as area overhead is also presented and compared to the ORBGRAND with a non-segmented sorter approach.

\subsection{ORBGRAND baseline implementation } \label{sec:ImplementationBaseline}

The proposed ORBGRAND VLSI architecture with parameters LW$\leq$64 and P$\leq$6 has been implemented in Verilog HDL and synthesized using Synopsys Design Compiler with general-purpose TSMC 65 nm CMOS technology. The design has been verified using test benches generated via the bit-true C model of the proposed hardware. Table \ref{table:tableORBGRAND_128} presents the synthesis results for the proposed decoder with $n=128$ and the proposed architecture can support code rates between $0.75$ and $1$. Input channel LLRs are quantized on 5 bits, including 1 sign bit and 3 bits for the fractional part. To ensure accuracy in power measurements, switching activities from real test vectors are extracted for all of the VLSI architectures presented in Table \ref{table:tableORBGRAND_128}.

The maximum frequency supported by the ORBGRAND implementation is $454~\text{MHz}$. Since there is no pipelining technique for the decoder core, one clock cycle corresponds to one time-step. The average decoding latency of the proposed hardware is calculated using the bit-true C model after taking account for at least 100 frames in error for each $\frac{E_b}{N_0}$ point. At target FER of $10^{-7}$ ($\frac{E_b}{N_0}$$>7.5$ dB), the average latency is only $2.47$ns, resulting in an average decoding information throughput of $42.5$ Gbps for a (128,105) polar code.  However, the worst-case (W.C.) scenario needs $\numprint{4226}$ cycles with $n= 128$ and parameters LW$\leq$64 and P$\leq$6, culminating in a W.C. latency of 9.3$\mu$s.

As seen in Fig. \ref{fig:fer_polar_bch} (a), ORBGRAND with parameters $LW\leq64$ and $P\leq6$ has similar decoding performance (target FER of $10^{-5}$) to the Dynamic SC-Flip (DSCF) \cite{Chandesris} decoder for (128,105) 5G-NR CRC-Aided (CA) polar code. The proposed ORBGRAND VLSI implementation ($LW\leq64$ and $P\leq6$) is compared to VLSI architecture for DSCF polar code decoder ($\omega = 2$, $T_\text{max}=50$) \cite{Ercan_2020}, which employs 7 and 6 bit internal and channel LLR quantizations, respectively. Compared to DSCF \cite{Ercan_2020}, ORBGRAND (LW$\leq$64, P$\leq$6) has a $8\times$ area overhead, as well as a $52\%$ increase in the worst-case latency. However, at a target FER of $10^{-7}$, the proposed ORBGRAND results in $49\times$ higher average throughput than the DSCF \cite{Ercan_2020}. Furthermore, as compared to DSCF \cite{Ercan_2020}, ORBGRAND (LW$\leq$64, P$\leq$6) is $5\times$ more area efficient and $32\times$ more energy efficient. Moreover, the proposed ORBGRAND hardware is code and rate compatible, while the DSCF \cite{Ercan_2020} can only decode polar codes.

In comparison to the hard-input GRANDAB decoder (AB=3) \cite{GRANDAB-VLSI}, ORBGRAND (LW$\leq$64, P$\leq$6) has a $7\times$ area overhead, as well as a $13.5\%$ higher W.C. and a $23.5\% $ higher average latency. 
Furthermore, as compared to \cite{GRANDAB-VLSI}, ORBGRAND (LW$\leq$64, P$\leq$6) is $2\times$ less energy efficient and $9\times$ less area efficient. However, as seen in Fig. \ref{fig:fer_polar_bch}, the FER performance of ORBGRAND (LW$\leq$64, P$\leq$6), a soft decision decoder, outperforms hard-input counterpart decoders by at least $1.3\sim2$dB for target FERs $\leq10^{-5}$.

\subsection{Design expansion and latency analysis }

As illustrated in Fig. \ref{fig:fer_polar_bch}, the ORBGRAND with parameters $LW\leq96$ and $P\leq8$ has similar decoding performance to the ORBGRAND with parameters $LW_\text{max} = 8256$ and $LW_\text{max} = 8128$ ($LW_\text{max} \leq \frac{n(n+1)}{2}$ \cite{duffy2020ordered} ) for both $(128,105)$ Polar code and $(127,106)$ BCH code. Furthermore, at a target FER of $\leq10^{-7}$, ORBGRAND with parameters $LW\leq96$ and $P\leq8$ results in a $0.2\sim0.3$dB gain in decoding performance when compared to ORBGRAND with parameters $LW\leq64$ and $P\leq6$.

Table \ref{table:tableORBGRAND_128} presents the VLSI implementation results for the proposed ORBGRAND VLSI architecture with parameters $LW\leq96$ and $P\leq8$ using the same implementation settings described in section \ref{sec:ImplementationBaseline}. The proposed ORBGRAND can support a maximum frequency of $454~\text{MHz}$. As shown in Table \ref{table:tableORBGRAND_128}, the ORBGRAND implementation with parameters $LW\leq96$ and $P\leq8$  incurs a $23.6\%$ area overhead when compared to the ORBGRAND implementation with parameters $LW\leq64$ and $P\leq6$. Furthermore, the ORBGRAND implementation with parameters $LW\leq96$ and $P\leq8$ is $18.8\%$ less area efficient and $27.7\%$ less energy efficient than the ORBGRAND implementation with parameters $LW\leq64$ and $P\leq6$.

The ORBGRAND parameters $LW$ and $P$ influence the decoding performance as well as the worst-case decoding latency of the proposed ORBGRAND VLSI hardware. In the worst-case scenario, the ORBGRAND with parameters $LW\leq64$ and $P\leq6$ requires $\numprint{4226}$ cycles ($n= 128$), whereas the ORBGRAND with parameters $LW\leq96$ and $P\leq8$ requires $\numprint{93417}$ cycles, resulting in a worst-case latency of 205.76$\mu$s. Figure \ref{fig:lat_tp} (a) depicts the worst-case latency (in clock cycles (\ref{eq:nb_steps})) of the proposed ORBGRAND hardware for various $LW$ and $P$ parameter values. Fig. \ref{fig:lat_tp} (b) depicts the information throughout corresponding to $k = 105$ and a maximum frequency of $454~\text{MHz}$ for the proposed ORBGRAND ($LW, P$) hardware.

To conclude, the ORBGRAND implementation parameters $(LW,P)$ can be appropriately chosen to strike a balance between area overhead, energy budget, and decoding performance requirements for a target application.

\subsection{ORBGRAND area optimization } 
In this section, we investigate the sorter module in the proposed ORGRAND $(LW,P)$ VLSI implementation and propose segmenting the sorter module into multiple partitions to reduce the area overhead of ORBGRAND hardware. The ORBGRAND decoding procedure, as described in Algorithm \ref{alg:ORBgrand}, begins by sorting channel LLRs ($\bm{y}$) in ascending order of their absolute value ($\vert\bm{y}\vert$). Any sorter \cite{sortingAlgo} (\textit{insertion sorter}, \textit{merge sorter}, \textit{bubble sorter}) may be used to sort $\vert\bm{y}\vert$ for the proposed ORBGRAND hardware. The sorter choice is determined by the target application's budget in terms of decoding latency and hardware overhead. On one end of the spectrum, we have sequential sorters with high latency but low hardware implementation cost, while on the other, we have parallel sorters with low latency but high hardware implementation cost.

The proposed ORBGRAND VLSI implementation employs a \textit{bitonic sorter} \cite{batcher68} of length $n$ that is pipelined to $\log_2(n)$ stages. As a result, the sorting procedure requires just $\log_2(n)$ clock cycles. The \textit{bitonic sorter} module of the proposed ORBGRAND VLSI implementation can be partitioned into multiple segments to reduce the ORBGRAND implementation's area overhead. The size and number of partitions influence ORBGRAND decoding performance as well as the area overhead of the proposed ORBGRAND VLSI implementation.
A \textit{bitonic sorter} module of length $n$ is segmented into $S$ segments, each having a size $\frac{n}{S}$, for the proposed segmented sorter approach. Please note that the number of segments $S$ should be chosen in such a way that the size of each segment $\frac{n}{S}$ is an integer. The segmented \textit{bitonic sorter} is depicted in Fig. \ref{fig:segSorter}, which employs four sorters ($S=4$) of length $\frac{n}{4}$, each of which receives an unique subset of channel LLRs ($\bm{y}$). To generate the final LLRs, the sorted LLRs from individual sorters are concatenated. The first four elements of the final sorted LLRs are comprised of the first element of the output of each sorter. Similarly, the second element of each sorter's output occupies the following four positions of the final sorted LLRs. This procedure is continued until the last elements of each sorter's output are placed in the last four positions of the sorted LLRs, as shown in Fig. \ref{fig:segSorter}.

A non-segmented sorter will sort the LLR elements $\vert\bm{y}_i\vert$ ($\forall i \in [1,n]$) to their correct location. However, compared to a non-segmented sorter, the sorted LLRs using a segmented sorter ($S>1$) will have elements that are displaced from their correct locations. To investigate the effect of segments on the displacement of LLR elements from their correct locations, we performed Monte-Carlo simulations and measured the percentage of LLR elements that were within a specified distance of their correct location. Table \ref{table:tableSorter} compares the displacement of LLR elements from their correct locations with varying numbers of segments employed in the segmented-sorter approach. Table \ref{table:tableSorter} shows that as the number of segments decreases, more elements are concentrated closer to their correct locations, but as the number of segments increases, LLR elements are concentrated further away from their correct locations.

Fig. \ref{fig:fer_polar_seg_sort} depicts the FER performance of implementing a segmented sorter approach for ORBGRAND (LW$\leq$96, P$\leq$8) decoding of (128,105) polar code. As seen in the Figure \ref{fig:fer_polar_seg_sort}, the ORBGRAND using the segmented sorter with $S=2$ and $S=4$ suffers from a FER performance degradation of $0.1$dB and $0.3$dB respectively, at the target FER of $10^{-6}$, as compared to ORBGRAND with a non-segmented sorter. Table \ref{table:tableORBGRAND_128} compares VLSI implementation results for the ORBGRAND (LW$\leq$96, P$\leq$8) using the non-segmented sorter to the proposed ORBGRAND (LW$\leq$96, P$\leq$8) using segmented sorter approach. As shown in Table \ref{table:tableORBGRAND_128}, the proposed ORBGRAND (LW$\leq$96, P$\leq$8) with non-segmented sorter incurs an area overhead of $8\%$ and $21.6\%$, respectively, when compared to ORBGRAND with segmented sorter parameters $S=2$ and $S=4$. To conclude, the number of sorter segments influences both decoding performance and the area overhead of the ORBGRAND hardware; they can be chosen appropriately to strike a balance between decoding performance requirements and area overhead for a target application.

\section{Conclusion}
In this work, we present a hardware architecture for the ORBGRAND algorithm. ORBGRAND is a soft input GRAND variant that generates test error patterns in a fixed logistic weight order, rendering it suitable for parallel hardware implementation. Due to the code-agnostic nature of the GRAND and its variants, the proposed ORBGRAND architecture can decode any code as long as the length and rate constraints are met. We suggest modifications in the ORBGRAND algorithm to simplify the hardware implementation and reduce the decoding complexity.  Furthermore, the proposed ORBGRAND VLSI architecture uses parameters that can be tweaked to meet the optimal decoding performance as well as the decoding latency for a specific application. According to ASIC synthesis results, an average decoding throughput of $42.5$ Gbps can be achieved for a code length of $128$ and a target FER of $10^{-7}$. The proposed VLSI architecture improves decoding performance by at least $2$ dB over the GRANDAB, a hard-input variant of GRAND. In comparison to the state-of-the-art DSCF hardware decoder for 5G $(128,105)$ polar code, the proposed VLSI implementation achieves $49\times$ higher decoding throughput, $32\times$ higher energy efficiency and $5\times$ higher area efficiency. Finally, the proposed architecture is the first step toward implementing GRAND family soft-input decoders in hardware.


%

\appendices
\section{PROOF OF LEMMA 1}
\begin{proof}
 It is sufficient to show that for all $i$ ($i \in [2, P]$) \\ $\lambda_i^\text{max} < \frac{2\times m - (i\times(i-1))+2-2\times\sum\limits_{j=i+1}^{P}\lambda_{j}}{2\times i}$. We use induction on $i$.

$\text{Base case ($i=2$)}$: 

\begin{align*} 
\lambda_2^\text{max} < \frac{m - \sum\limits_{j=3}^{P}\lambda_{j}}{2} \\
\end{align*}
$\text{Since the $\lambda_i$ are ordered $\therefore$ $\lambda_2<\lambda_1$  }$
\begin{align*} 
\Rightarrow \lambda_2^\text{max} &< m-\sum\limits_{j=2}^{P}\lambda_{j} \text{ } (\because\lambda_1=m-\sum\limits_{j=2}^{P}\lambda_{j}) \\
\Rightarrow \lambda_2^\text{max} &< m-\sum\limits_{j=3}^{P}\lambda_{j}-\lambda_2^\text{max} \\
\therefore \lambda_2^\text{max} &< \frac{m - \sum\limits_{j=3}^{P}\lambda_{j}}{2}
\end{align*}
$\text{Inductive hypothesis ($i=k$)}$: 
\begin{align*} 
\lambda_k^\text{max} < \frac{2\times m - (k\times(k-1))+2-2\times\sum\limits_{j=k+1}^{P}\lambda_{j}}{2\times k}
\end{align*}
$\text{Inductive step ($i=k+1$)}$: 
\begin{align*} 
\lambda_{k+1}^\text{max} < \frac{2\times m - (k\times(k+1))+2-2\times\sum\limits_{j=k+2}^{P}\lambda_{j}}{2\times(k+1)}
\end{align*}
$\because$ $\text{Inductive hypothesis }$ 
\begin{align*} 
\lambda_k^\text{max} &< \frac{2\times m - (k\times(k-1))+2-2\times\sum\limits_{j=k+1}^{P}\lambda_{j}}{2\times k} 
\end{align*}
\begin{align*} 
\Rightarrow \lambda_k^\text{max} &< \frac{2\times m - (k\times(k+1))+2-2\times\sum\limits_{j=k+1}^{P}\lambda_{j}}{2\times k} + 1 
\end{align*}
\begin{align*} 
\text{ } (\because\text{}k\times(k-1)=k\times(k+1)+2\times(k))
\end{align*}
\begin{align*} 
\Rightarrow \lambda_k^\text{max} &< \frac{2\times m - (k\times(k+1))+2-2\times\sum\limits_{j=k+2}^{P}\lambda_{j}}{2\times k} \\&+ 1 - \frac{\lambda_{k+1}^\text{max}}{k} 
\end{align*}
\begin{align*} 
\text{ } (\because\sum\limits_{j=k+1}^{P}\lambda_{j}=\lambda_{k+1}^\text{max} + \sum\limits_{j=k+2}^{P}\lambda_{j})
\end{align*}

\begin{align*} 
\Rightarrow &\frac{\lambda_{k+1}^\text{max}+\lambda_k^\text{max} - k}{k} \\ &< \frac{2\times m - (k\times(k+1))+2-2\times\sum\limits_{j=k+2}^{P}\lambda_{j}}{2\times k}\\
\end{align*}
\begin{align*} 
\Rightarrow &\frac{\lambda_{k+1}^\text{max}+\lambda_k^\text{max} - k}{k+1} \\ &< \frac{2\times m - (k\times(k+1))+2-2\times\sum\limits_{j=k+2}^{P}\lambda_{j}}{2\times (k+1)} \labeln{eq:L1} \\
\end{align*}
\begin{align*} 
\text{Since $\lambda_i$ are ordered} \therefore \lambda_{k+1} &< \lambda_{k} 
\end{align*}
\begin{align*} 
\Rightarrow \lambda_{k+1}^\text{max}+1 &\leq \lambda_{k}^\text{max} \text{  } (\because \lambda_{k+1}^\text{max}+1 \leq \lambda_{k} \leq \lambda_{k}^\text{max} ) 
\end{align*} 
\begin{align*} 
\Rightarrow  k\times\lambda_{k+1}^\text{max}+k + \lambda_{k+1}^\text{max} &\leq k\times\lambda_{k}^\text{max} + \lambda_{k+1}^\text{max}   
\end{align*}
\begin{align*} 
\Rightarrow \lambda_{k+1}^\text{max} &\leq \frac{\lambda_{k+1}^\text{max}+\lambda_k^\text{max} - k}{k+1} 
\end{align*}
\begin{align*} 
\text{ Using \eqref{eq:L1}: } &\lambda_{k+1}^\text{max} \leq\frac{\lambda_{k+1}^\text{max}+\lambda_k^\text{max} - k}{k+1} 
\end{align*}
\begin{align*} 
 &< \frac{2\times m - (k\times(k+1))+2-2\times\sum\limits_{j=k+2}^{P}\lambda_{j}}{2\times (k+1)}
\end{align*}
\begin{align*} 
\Rightarrow  &\lambda_{k+1}^\text{max} < \frac{2\times m - (k\times(k+1))+2-2\times\sum\limits_{j=k+2}^{P}\lambda_{j}}{2\times (k+1)}
\end{align*}
\begin{align*} 
(\because &a \leq b < c \Rightarrow a < c) 
\end{align*}

\end{proof}

\ifCLASSOPTIONcaptionsoff
  \newpage
\fi



%



\bibliographystyle{IEEEbib}
\bibliography{refs}

\end{document}